\documentclass[12pt,draftcls,onecolumn]{IEEEtran}
\usepackage{latexsym}
\usepackage{graphicx}
\usepackage{array}
\usepackage{amsmath}
\usepackage{amsfonts}
\usepackage{amssymb}
\usepackage{amsthm}
\usepackage{algpseudocode}
\usepackage{algorithm}

\newtheorem{thm}{Theorem}
\newtheorem{lemma}{Lemma}
\newtheorem{prop}{Proposition}

\newtheorem{cor}{Corollary}

\newcommand{\bx} {\boldsymbol{x}}

\newcommand{\by} {\boldsymbol{y}}

\newcommand{\bH} {\boldsymbol{H}}
\newcommand{\bB} {\boldsymbol{B}}
\newcommand{\bA} {\boldsymbol{A}}

\newcommand{\bI} {\boldsymbol{I}}
\newcommand{\bR} {\boldsymbol{R}}
\newcommand{\bU} {\boldsymbol{U}}
\newcommand{\bu} {\boldsymbol{u}}

\newcommand{\bW} {\boldsymbol{W}}
\newcommand{\bQ} {\boldsymbol{Q}}
\newcommand{\bM} {\boldsymbol{M}}
\newcommand{\bP} {\boldsymbol{P}}

\newcommand{\bLam} {\boldsymbol{\Lambda}}
\newcommand{\gl}{\lambda}

\newcommand{\bxi} {\boldsymbol{\xi}}

\newcommand{\sN} {\mathcal{N}}
\newcommand{\sR} {\mathcal{R}}
\newcommand{{\diag}} {\mathrm{diag}}

\def\bal#1\eal{\begin{align}#1\end{align}}
\newcommand{\bp} {\begin{proof}}
\newcommand{\ep} {\end{proof}}

\newcommand{{\bRF}} {\right\}}

\newcommand{\tr}{\operatorname{tr}}

\begin{document}

\title{The Capacity and Optimal Signaling for Gaussian MIMO Channels Under Interference Constraints (full version)}

\author{Sergey Loyka, Senior Member, IEEE

\vspace*{-1\baselineskip}

\thanks{This paper was presented in part at the 5th IEEE Global Conference on Signal and Information Processing, Montreal, Canada, Nov. 2017 \cite{Loyka-17-2}.}

\thanks{S. Loyka is with the School of Electrical Engineering and Computer Science, University of Ottawa, Ontario, Canada, e-mail: sergey.loyka@uottawa.ca.}

}

\maketitle

\vspace*{-1\baselineskip}
\begin{abstract}
Gaussian MIMO channel under total transmit and interference power constraints (TPC and IPC) is considered. A closed-form solution for the optimal transmit covariance matrix in the general case is obtained using the KKT-based approach (up to dual variables). While closed-from solutions for optimal dual variables are possible in special cases, an iterative bisection algorithm (IBA) is proposed to find the optimal dual variables in the general case and its convergence is proved for some special cases. Numerical experiments illustrate its efficient performance. Bounds for the optimal dual variables are given, which facilitate numerical solutions. An interplay between the TPC and IPC is studied, including the transition from power-limited to interference-limited regimes as the total transmit power increases. Sufficient and necessary conditions for each constraint to be redundant are given.  A number of explicit closed-form solutions are obtained, including full-rank and rank-1 (beamforming) cases as well as the case of identical eigenvectors (typical for massive MIMO settings). A bound on the rank of optimal covariance is established. A number of unusual properties of optimal covariance matrix are pointed out.
\end{abstract}

\vspace*{-1.5\baselineskip}

\section{Introduction}

Growing volume of high-rate mobile wireless traffic stimulated active development of 5G standards and systems. Due to very high expectations, several new key technologies have been identified to meet those demands, including massive MIMO, millimeter waves (mmWave) and non-orthogonal multiple-access \cite{Shafi-17}. The ultimate goal is to increase significantly the available bandwidth as well as spectral efficiency to meet the growing traffic demands. However, aggressive frequency re-use and non-orthogonal access schemes can potentially generate significant amount of inter-user interference, which thus has to be carefully managed \cite{Liang-17}-\cite{Song-17}. This is somewhat similar to cognitive radio (CR) systems, which also emerged as a powerful approach to exploit underutilized spectrum and hence possibly resolve the spectrum scarcity problem \cite{Haykin-09}. In both settings, allowing spectrum re-use calls for a careful management of possible interference. In this respect, multi-antenna (MIMO) systems have significant potential due to their significant signal processing capabilities, including interference cancellation and precoding \cite{Biglieri}, which can also be done in an adaptive and distributed manner \cite{Scurati-10}. A promising approach is to limit interference to primary receivers (PR) by properly designing secondary transmitters (Tx) while exploiting their  multi-antenna capabilities.

The capacity and optimal signalling for the Gaussian MIMO channel under the total power constraints (TPC) is well-known: the optimal (capacity-achieving) signaling is Gaussian and, under the TPC, is on the eigenvectors of the channel with power allocation to the eigenmodes given by the water-filling (WF) \cite{Biglieri}-\cite{Telatar-95}. Under per-antenna power constraints (PAC), in addition or instead of the TPC, Gaussian signalling is still optimal but not on the channel eigenvectors anymore so that the standard water-filling solution over the channel eigenmodes does not apply \cite{Vu-11}\cite{Loyka-17}. Much less is known under the added interference power constraint (IPC), which limits the power of interference induced by the secondary transmitter to a primary receiver. A game-theoretic approach to this problem was proposed in \cite{Scurati-10}, where a fixed-point equation was formulated from which the optimal covariance matrix can in principle be determined. Unfortunately, no closed-form solution is known for this equation. In addition, this approach is limited in the following respects: the channel to the primary receiver is required to be full-rank (hence excluding the important case of single-antenna devices communicating to a multi-antenna base station or, in general, the cases where the number of Rx antennas is less than the number of Tx antennas - typical for massive MIMO downlink); the TPC is not included explicitly (rather, being "absorbed" into the IPC), hence eliminating the important case of inactive IPC (since this is the only explicit constraint); consequently, no interplay between the TPC and the IPC can be studied.

Earlier studies on cognitive radio MIMO system optimization under interference constraint using game-theoretic approach are extended to the case of channel uncertainty in \cite{Yang-13} by developing a number of numerical algorithms for Tx optimization under global interference constraints.  Due to the non-convex nature of the original problem, a number of approximate and sub-optimal approaches are adopted, for which provable convergence to a global optimum is out of reach. No closed-form solutions are known for this setting either. Weighted sum-rate maximization in multiuser MISO channel is considered in \cite{Huh-10} under interference constraints and numerically-efficient algorithms for Tx optimization based on zero-forcing beamforming are developed. No closed-form solutions are obtained for this problem. Gaussian MIMO broadcast (BC) and multiple-access channels (MAC) are considered in \cite{Zhang-12} under general linear Tx covariance constraint, which can also be interpreted as interference constraint, and the earlier BC-MAC duality result is extended to this more general setting. However, no closed-form solutions are obtained for an optimal Tx covariance matrix. In the related context of physical-layer security, the secrecy capacity of the Gaussian MIMO wiretap channel under interference constraints has been characterized in \cite {Dong-18} as a non-convex maximization  problem (over feasible Tx covariance matrices under the TPC and IPC) or as a convex-concave max-min problem (where "min" is over noise covariance matrices of a genie-aided channel), for which no closed-form solution is known in general.

Unlike most of the studies above, we concentrate here on analysis of the problem and obtain a number of closed-form solutions, which are validated via numerical experiments. This provides deeper understanding of the problem and a number of insights unavailable from numerical algorithms alone. Specifically, we obtain closed-form solutions for an optimal covariance matrix of the Gaussian MIMO channel under the TPC and the IPC using the KKT conditions. Both constraints are included explicitly and hence anyone is allowed to be inactive. This allows us to study the interplay between the power and interference constraints and, in particular, the transition from power-limited to interference limited regimes as the Tx power increases. As an added benefit, no limitations is placed on the rank of the channel to the PR, so that the number of antennas of the PR can be any. Under the added IPC, independent signaling is shown to be sub-optimal for parallel channels to the intended receiver (Rx), unless the PR channels are also parallel or if the IPC is inactive. These results are also extended to multiple IPCs.

Optimal signaling for the Gaussian MIMO channel under the TPC and the IPC has been also studied in \cite{Zhang-10}\cite{Yu-10} using the dual problem approach, and was later extended to multi-user settings in \cite{Liu-12}. However, constraint matrices are required to be full-rank and no closed-form solution was obtained for optimal dual variables. Hence, various numerical algorithms or sub-optimal solutions were proposed (e.g. partial channel projection). Similar problem has been also considered in \cite{Pham-18} under multiple linear constraints at the transmitter and an iterative numerical algorithm was developed via a min-max reformulation of the problem. However, no closed-form solution was obtained and the properties of optimal signaling as well as those of the capacity were not studied. Here, a closed-form solution for an optimal covariance matrix is obtained in the general case (up to dual variables), its properties are studied and a number of more explicit solutions are obtained in some special cases (including explicit solutions for dual variables). Our KKT-based approach does not require full-rank constraint matrices and includes explicit equations for the optimal dual variables, which can be solved efficiently. To this end, we propose an iterative (gradient-free) bisection algorithm (IBA) and prove its convergence. Numerical experiments demonstrate its efficient performance. In some cases, our KKT-based approach leads to closed-form solutions for the optimal dual variables, including full-rank and rank-1 (beamforming) solutions. Bounds to the optimal dual variables are derived, which facilitate numerical solutions. Properties of the optimal Tx covariance as a function of dual variables are explored: the total Tx power as well as interference power are shown to be decreasing functions of dual variables, which is an important part in the proof of the IBA convergence.

The above solutions for optimal covariance posses a number of unusual properties (not found in the standard WF procedure), namely: an optimal covariance is not necessarily unique; its rank can exceed the main channel rank; the TPC can be inactive; signaling on the main channel eigenmodes is not optimal anymore. Ultimately, these are due to an interplay between the TPC and the IPC.

A simple rank condition is given to characterize the cases where spectrum sharing is possible for any interference power constraint. In general, the primary user has a major impact on the capacity at high SNR while being negligible at low SNR. The high-SNR behaviour of the capacity is qualitatively determined by the null space of the PR's channel matrix.

The presented closed-form solutions of optimal signaling can be used directly in massive MIMO settings. Since numerical complexity of generic convex solvers can be prohibitively large for massive MIMO (in general, it scales as $m^6$ with the number $m$ of antennas), the above analytical solutions are a valuable low-complexity alternative.

In all considered cases, optimal Tx covariance matrices are significantly different from those of the standard Gaussian MIMO channel under the TPC and/or the PAC \cite{Vu-11}\cite{Loyka-17}, and from those of the wiretap channel in \cite{Dong-18}. In the latter two cases, optimal Tx covariance matrix remains unknown in the general case while some special cases have been solved.

Finally, it should be pointed out that the channel model we consider here (the standard point-to-point Gaussian MIMO channel under interference constraints at the transmitter) is different from the Gaussian interference channel (G-IC) where multi-user interference is present at each receiver and there is no interference constraint at the transmitters, as in e.g. \cite{Shang-13}-\cite{Lagen-16}. In the latter case, the capacity and optimal signaling are not known in general, even for the 2-user SISO channel \cite{Shang-13} (it is not even known whether Gaussian signaling is optimal in general), except for some special cases, such as strong and weak interference regimes \cite{Sato-81}\cite{Annapureddy-11}, so that various bounds \cite{Etkin-08}\cite{Motahari-09} and ad-hoc signaling techniques \cite{Ye-03}-\cite{Lagen-16} are used instead. On the contrary, optimal signaling is known to be Guassian for the channel model considered here and the capacity can be expressed as an optimization problem over all feasible transmit covariance matrices. An analytical solution of this problem in the general case and its properties are the main contributions of the present paper.

\textit{Notations}: bold capitals ($\bR$) denote matrices while bold lower-case letters ($\bx$) denote column vectors; $\bR^+$ is the Hermitian conjugation of $\bR$; $\bR \ge 0$ means that $\bR$ is positive semi-definite; $|\bR|,\ tr(\bR),\ r(\bR)$ denote determinant, trace and rank of $\bR$, respectively; $\gl_i(\bR)$ is $i$-th eigenvalue of $\bR$; unless indicated otherwise, eigenvalues are in decreasing order, $\gl_1\ge \gl_2\ge ..$; $\lceil\cdot\rceil$ denotes ceiling, while $(x)_+=\max[0,x]$ is the positive part of $x$; $\mathcal{R}(\bR)$ and $\mathcal{N}(\bR)$ denote the range and null space of $\bR$ while $\bR^{\dag}$ is its Moore-Penrose pseudo-inverse; $\mathbb{E}\{\cdot\}$ is statistical expectation.

\section{Channel Model}
\label{sec.Channel Model}

Let us consider the standard discrete-time model of the Gaussian MIMO channel:
\bal
\label{eq.ch.mod}
\by_1 = \bH_1\bx +\bxi_1
\eal
where $\by_1, \bx, \bxi_1$ and $\bH_1$ are the received and transmitted signals, noise and channel matrix, of dimensionality $n_1\times 1$, $m\times 1$, $n_1\times 1$, and $n_1\times m$, respectively, where $n_1, m$ are the number of Rx and Tx antennas. This is illustrated in Fig. 1. The noise is assumed to be complex Gaussian with zero mean and unit variance, so that the SNR equals to the signal power. A complex-valued channel model is assumed throughout the paper, with full channel state information available both at the transmitter and the receiver. Gaussian signaling is known to be optimal in this setting \cite{Biglieri}-\cite{Telatar-95} so that finding the channel capacity $C$ amounts to finding an optimal transmit covariance matrix $\bR$, which can be expressed as the following optimization problem (P1):
\bal
\label{eq.C.def}
(P1):\ C = \max_{\bR \in S_R} C(\bR)
\eal
where $C(\bR) = \log|\bI +\bW_1\bR|$, $\bW_1=\bH_1^+\bH_1$, $\bR$ is the Tx covariance and $S_R$ is the constraint set. In the case of the total power constraint (TPC) only, it takes the form
\bal
S_R=\{\bR: \bR\ge 0, \tr(\bR) \le P_T\},
\eal
where $P_T$ is the maximum total Tx power. The solution to this problem is well-known: optimal signaling is on the eigenmodes of $\bW_1$, so that they are also the eigenmodes of optimal covariance $\bR^*$, and the optimal power allocation is via the water-filling (WF). This solution can be compactly expressed as follows:
\bal
\label{eq.RWF}
\bR^*=\bR_{WF}\triangleq(\mu^{-1}\bI-\bW_1^{-1})_+
\eal
where $\mu\ge 0$ is the "water" level found from the total power constraint $tr(\bR^*)=P_T$ and $(\bR)_+$ denotes positive eigenmodes of Hermitian matrix $\bR$:
\bal
\label{eq.R+}
(\bR)_+=\sum_{i: \gl_i>0} \gl_i\bu_i\bu_i^+
\eal
where $\gl_i,\ \bu_i$ are $i$-th eigenvalue and eigenvector of $\bR$. Note that this definition allows $\bW_1$ in \eqref{eq.RWF} to be singular, since the respective summation (as in \eqref{eq.R+}) includes only strictly-positive modes: $\bR_{WF} = \sum_{i: \gl_{1i} >\mu} (\mu^{-1} -\gl_{1i}^{-1})\bu_{1i}\bu_{1i}^+$, where $\gl_{1i}, \ \bu_{1i}$ are $i$-th eigenvalue and eigenvector of $\bW_1$, so that $\gl_{1i} =0$ are excluded from the summation due to $\gl_{1i} >\mu \ge 0$.

In a multi-user system, there is a 2nd channel from the Tx to an external (primary) user U, see Fig. \ref{fig.system},
\bal
\by_2=\bH_2\bx+\bxi_2
\eal
where $\bH_2$ is the matrix of the Tx-U channel (see Fig. 1); all vector/matrix dimensions are equal to the respective number of antennas. There is a limit on how much interference the Tx can induce (via $\bx$) to the user U:
\bal
\mathbb{E}\{\bx^+\bH_2^+\bH_2\bx\} = tr(\bH_2\bR\bH_2^+) \le P_I
\eal
where $P_I$ is the maximum acceptable interference power and the left-hand side is the actual interference power at user U. In this setting, the constraint set becomes
\bal
\label{eq.SR.2}
S_R=\{\bR: \bR\ge 0,\ tr(\bR) \le P_T,\ tr(\bW_2\bR) \le P_I\},
\eal
where $\bW_2=\bH_2^+\bH_2$; this is extended to multiple IPCs in Section \ref{sec.multi}. The IPC in \eqref{eq.SR.2} is consistent with the respective per-user constraints in the CR setting in \cite{Scurati-10}-\cite{Zhang-10}, with the relay system setting in \cite{Yu-10} as well as with the wiretap channel setting in \cite{Dong-18}. The Gaussian signalling is still optimal in this setting and the capacity subject to the TPC and IPC can still be expressed as in \eqref{eq.C.def} but the optimal covariance is not $\bR_{WF}$ anymore, as discussed in the next section.

\begin{figure}[t]
	\centerline{\includegraphics[width=3.5in]{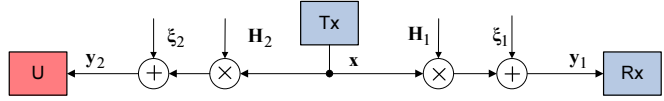}}
	\caption{A block diagram of Gaussian MIMO channel under interference constraint. $\bH_1$ and $\bH_{2}$ are the channel matrices to the Rx and an external (primary) user U respectively.}
		\label{fig.system}
\end{figure}

\section{Optimal Signalling Under Interference Constraint}
\label{sec.C.Int}

The following Theorem gives a closed-form solution for the optimal Tx covariance matrix under the TPC and the IPC in \eqref{eq.SR.2} in the general case.

\begin{thm}
\label{thm.R*}
Consider the capacity of the Gaussian MIMO channel in \eqref{eq.C.def} under the joint TPC and IPC in \eqref{eq.SR.2},
\bal
C = \max_{\bR} C(\bR)\ \mbox{s.t.}\ \bR\ge 0,\ tr(\bR) \le P_T,\ tr(\bW_2\bR) \le P_I
\eal
The optimal Tx covariance matrix to achieve the capacity can be expressed as follows:
\bal
\label{eq.thm.R*.1}
\bR^* = \bW_{\mu}^{\dag} (\bI-\bW_{\mu}\bW_1^{-1}\bW_{\mu})_+ \bW_{\mu}^{\dag}
\eal
where $\bW_{\mu}=(\mu_1\bI+\mu_2\bW_2)^{\frac{1}{2}}$; $\bW_{\mu}^{\dag}$ is Moore-Penrose pseudo-inverse of $\bW_{\mu}$; $\mu_1, \mu_2\ge 0$ are Lagrange multipliers (dual variables) responsible for the total Tx and interference power constraints found as solutions of the following non-linear equations:
\bal
\label{eq.thm.R*.2}
\mu_1(tr(\bR^*) - P_T)=0,\ \mu_2(tr(\bW_2\bR^*)-P_I)=0
\eal
subject to $tr(\bR^*)\le P_T$, $tr(\bW_2\bR^*) \le P_I$. The capacity can be expressed as follows:
\bal
\label{eq.C}
C = \sum_{i: \gl_{ai}>1} \log \gl_{ai}
\eal
where $\gl_{ai} = \gl_i(\bW_{\mu}^{\dag}\bW_1\bW_{\mu}^{\dag})$.
\end{thm}
\begin{proof}
See Appendix.
\end{proof}

As indicated in Section \ref{sec.multi}, this solution can be easily extended to multiple IPCs. When $\bW_{\mu}$ is full-rank, $\bW_{\mu}^{\dag}=\bW_{\mu}^{-1}$ \cite{Horn-85} and $\bR^*$ in \eqref{eq.thm.R*.1} reduces to the respective solutions in \cite{Zhang-10}\cite{Yu-10}.

Next, we explore some general properties of the capacity. It is well-known that, without the IPC, $C(P_T)$ grows unbounded as $P_T$ increases, $C(P_T) \rightarrow \infty$ as $P_T \rightarrow \infty$ (assuming $\bW_1 \neq 0$). This, however, is not necessarily the case under the IPC with fixed $P_I$. The following proposition gives sufficient and necessary conditions when it is indeed the case.

\begin{prop}
\label{prop.C.inf}
Let $0 \le P_I < \infty$ be fixed. Then, the capacity grows unbounded as $P_T$ increases, i.e. $C(P_T) \rightarrow \infty$ as $P_T \rightarrow \infty$, if and only if $\sN(\bW_2) \notin \sN(\bW_1)$.
\end{prop}
\begin{proof}
See Appendix.
\end{proof}

Since the condition of this Proposition is both sufficient and necessary for the unbounded growth of the capacity, it gives the exhaustive characterization of all the cases where such growth is possible. In practical terms, those cases represent the scenarios where any high spectral efficiency is achievable given enough power budget. On the other hand, if $\sN(\bW_2) \in \sN(\bW_1)$, then very high spectral efficiency cannot be achieved even with unlimited power budget, due to the dominance of the IPC. It can be seen that the condition $\sN(\bW_2) \notin \sN(\bW_1)$ holds if $r(\bW_2) < r(\bW_1)$, and hence the capacity grows unbounded with $P_T$ under the latter condition.

In the standard Gaussian MIMO channel without the IPC, $C=0$ if either $P_T=0$ or/and $\bW_1=0$, i.e. in a trivial way. On the other hand, in the same channel under the TPC and IPC, the capacity can be zero in non-trivial ways, as the following proposition shows. In practical terms, this characterizes the cases where interference constraints of primary users rule out any positive rate of a given user and, hence, spectrum sharing is not possible.

\begin{prop}
Consider the Gaussian MIMO channel under the TPC and IPC and let $P_T>0$. Its capacity is zero if and only if $P_I=0$ and $\sN(\bW_2) \in \sN(\bW_1)$.
\end{prop}
\begin{proof}
To prove the "if" part, observe that $tr(\bW_2\bR)=P_I=0$ implies that $\bW_2\bR=0$ (since $\bW_2\bR$ has positive eigenvalues, $\gl_i(\bW_2\bR) \ge 0$) so that $\sR(\bR) \in \sN(\bW_2) \in \sN(\bW_1)$ and hence $\bW_1\bR=0$ so that $\log|\bI+\bW_1\bR|=0$ for any feasible $\bR$. Hence, $C=0$.

To prove the "only if" part, assume first that $P_I>0$ and set $\bR=p\bI$, where $p= \min\{P_T, P_I/(m\gl_1(\bW_2))\}$. Note that $\bR$ is feasible: $tr(\bR) \le P_T$ and $tr(\bW_2\bR) \le P_I$. Furthermore,
\bal
C \ge \log|\bI+\bW_1\bR| >0
\eal
and hence $P_I=0$ is necessary for $C=0$. To show that $\sN(\bW_2) \in \sN(\bW_1)$ is necessary as well, assume that $\sN(\bW_2) \notin \sN(\bW_1)$, which implies that $\exists \bu: \bW_2\bu=0, \bW_1\bu \neq 0$. Now set $\bR=P_T\bu\bu^+$, for which $tr(\bR) =P_T, tr(\bW_2\bR)=0$, so it is feasible and
\bal
C \ge \log|\bI+\bW_1\bR| = \log(1+P_T\bu^+\bW_1\bu) >0
\eal
so that $\sN(\bW_2) \in \sN(\bW_1)$ is necessary for $C=0$.
\end{proof}

Note that the condition $P_I=0$ is equivalent to zero-forcing transmission, i.e. the capacity is zero only if the ZF transmission is required; otherwise, $C>0$. The condition $\sN(\bW_2) \in \sN(\bW_1)$ cannot be satisfied if $r(\bW_1) > r(\bW_2)$ and hence $C>0$ under the latter condition, which is also sufficient for unbounded growth of the capacity with $P_T$. This is summarized below.

\begin{cor}
\label{cor.C.inf}
If $r(\bW_1) > r(\bW_2)$, then

1. $C \neq 0$ $\forall\ P_I \ge 0$ and $P_T>0$.

2. $C(P_T) \rightarrow\infty$ as  $P_T \rightarrow\infty$ $\forall\ P_I \ge 0$
\end{cor}

Thus, the condition $r(\bW_1) > r(\bW_2)$ represents favorable scenarios where spectrum sharing is possible for any $P_I$ and arbitrary large capacity can be attained given enough power budget.

It should be pointed out that \eqref{eq.thm.R*.2} in Theorem \ref{thm.R*} allows anyone of the dual variables to be inactive (i.e. $\mu_1=0$ or $\mu_2=0$, but not simultaneously), unlike the standard WF solution, where the TPC is always active. While it is not feasible to find dual variables $\mu_1, \mu_2$ in a closed form in general (since \eqref{eq.thm.R*.2} is a system of coupled non-linear equations), they can be found in such form in some special cases, as the next sections show. Section \ref{sec.IBA} will develop an iterative bisection algorithm (IBA) to find the optimal dual variables in the general case with any desired accuracy and prove its convergence.

\section{Full-rank solutions}

While Theorem \ref{thm.R*} establishes a closed-form solution for optimal covariance $\bR^*$ in the general case, it is expressed via dual variables $\mu_1, \mu_2$ for which no closed-form solution is known in general so they have to be found numerically using \eqref{eq.thm.R*.2}. This limits insights significantly. In this section, we explore the cases when the optimal covariance $\bR^*$ is of full rank and obtain respective closed-form solutions. First, we consider a transmit power-limited regime, where the IPC is redundant and hence the TPC is active.

\begin{prop}
\label{prop.FR.TP}
Let $\bW_1>0$ and $P_T$ be bounded as follows:
\bal
\label{eq.prop.FR.TP}
m\gl_1(\bW_1^{-1}) &- tr(\bW_1^{-1}) < P_T \le \frac{m}{tr(\bW_2)} (P_I+tr(\bW_2\bW_1^{-1})) - tr(\bW_1^{-1})
\eal
then $\mu_2=0$, i.e. the IPC is redundant\footnote{i.e. can be ommited without affecting the capacity, which corresponds to $\mu_2=0$.}, $\bR^*$ is of full-rank and is given by:
\bal
\label{eq.prop.FR.TP.2}
\bR^* = \mu_1^{-1}\bI - \bW^{-1}_1
\eal
where $\mu^{-1}_1=m^{-1}(P_T+tr(\bW^{-1}_1))$. The capacity can be expressed as
\bal
\label{eq.prop.FR.TP.3}
C= m\log((P_T+tr(\bW^{-1}_1))/m) +\log|\bW_1|
\eal
\end{prop}
\begin{proof}
It follows from \eqref{eq.thm.R*.1} that, if $\bR^*$ is full rank, then so is $\bW_{\mu}$ and
\bal
(\bI-\bW_{\mu}\bW_1^{-1}\bW_{\mu})_+ = \bI-\bW_{\mu}\bW_1^{-1}\bW_{\mu} >0
\eal
(since it is of full rank) and hence
\bal
\bR^* = \bW_{\mu}^{-2} - \bW_1^{-1}
\eal
The full-rank condition $\bR^*>0$ is equivalent to
\bal\notag
\label{eq.FR.P.1}
\bR^*>0\ &\leftrightarrow \bW_{\mu}^{-2} > \bW_1^{-1}\ \leftrightarrow \bW_1 > \bW_{\mu}^2 = \mu_1\bI + \mu_2\bW_2\\
 &\leftrightarrow \bW_1 - \mu_2\bW_2 > \mu_1\bI \ \leftrightarrow \gl_m(\bW_1 - \mu_2\bW_2) > \mu_1
\eal
where we used the standard tools of matrix analysis \cite{Horn-85}\cite{Zhang-99}. When the IPC is redundant, $\mu_2=0$ from which \eqref{eq.prop.FR.TP.2} and
\bal
\label{eq.RF.P.2}
\mu_1^{-1}=m^{-1}(P_T+tr(\bW^{-1}_1))
\eal
follow. It remains to establish \eqref{eq.prop.FR.TP}. To this end, $\mu_2=0$ implies $\mu_1>0$ (active TPC) and $\gl_m(\bW_1) > \mu_1$ (from \eqref{eq.FR.P.1}) together with \eqref{eq.RF.P.2} implies 1st inequality in \eqref{eq.prop.FR.TP}. 2nd inequality in \eqref{eq.prop.FR.TP} ensures that the IPC is redundant: $\mu_2=0$ and $tr(\bW_2\bR^*) \le P_I$. Indeed, the capacity under the joint constraints (TPC+IPC) cannot exceed that under the TPC alone, for which the optimal covariance $\bR^*$ is as in \eqref{eq.prop.FR.TP.2}. However, under 2nd inequality in \eqref{eq.prop.FR.TP}, $\bR^*$  satisfies $tr(\bW_2\bR^*) \le P_I$ and hence $\bR^*$ is feasible under the joint constraints as well so that (i) both capacities are equal and (ii) $\bR^*$ in \eqref{eq.prop.FR.TP.2} is also optimal under the joint constraints, which corresponds to $\mu_2=0$ (so that the IPC can be removed without affecting the capacity).  \eqref{eq.prop.FR.TP.3} is obtained by using \eqref{eq.prop.FR.TP.2} in $C(\bR)$.
\end{proof}

Next, we consider an interference-limited regime, where the TPC is redundant and hence the IPC is active.

\begin{prop}
Let $\bW_1, \bW_2 >0$ and $P_I$ be bounded as follows:
\bal
\label{eq.prop.FR.IP}
m\gl_1(\bW_2 &\bW^{-1}_1) - tr(\bW_2\bW^{-1}_1) < P_I \le \frac{m}{tr(\bW_2^{-1})} (P_T+tr(\bW^{-1}_1)) - tr(\bW_2\bW^{-1}_1)
\eal
then $\mu_1=0$, i.e. the TPC is redundant, $\bR^*$ is of full-rank and is given by:
\bal
\label{eq.prop.FR.IP.2}
\bR^* = \mu_2^{-1}\bW_2^{-1} - \bW^{-1}_1
\eal
where $\mu_2^{-1}=m^{-1}(P_I+tr(\bW_2\bW^{-1}_1))$. The capacity can be expressed as
\bal
\label{eq.prop.FR.IP.3}
C= m\log((P_I&+tr(\bW_2\bW_1^{-1}))/m) +\log|\bW_1|-\log|\bW_2|
\eal
\end{prop}
\begin{proof}
When the TPC is redundant, $\mu_1=0$ and \eqref{eq.prop.FR.IP.2} with
\bal
\mu_2^{-1}=m^{-1}(P_I+tr(\bW_2\bW^{-1}_1))
\eal
follow from \eqref{eq.thm.R*.1} in the same way as for Proposition \ref{prop.FR.TP}. 1st condition in \eqref{eq.prop.FR.IP} follows from $\bR^*>0$, which, using \eqref{eq.FR.P.1}, is equivalent to
\bal
\gl_m(\bW_2^{-1}\bW_1) > \mu_2= m(P_I+tr(\bW_2\bW^{-1}_1))^{-1}
\eal
2nd condition in \eqref{eq.prop.FR.IP} ensures that the TPC is redundant since the Tx power is sufficiently large: $\mu_1=0$ and $tr(\bR^*) \le P_T$.
\end{proof}

It is clear from \eqref{eq.prop.FR.TP.3} and \eqref{eq.prop.FR.IP.3} that the latter represents an interference-limited scenario while the former - a Tx power-limited one.

It follows from the proof of Proposition \ref{prop.FR.TP} that a general full-rank solution is given by
\bal
\label{eq.R*.general}
\bR^* = \bW_{\mu}^{-2} - \bW^{-1}_1
\eal
and the respective capacity is $C=\log|\bW_1|-2\log|\bW_{\mu}|$, which holds if $\gl_m(\bW_1-\mu_2\bW_2)>\mu_1$. If neither \eqref{eq.prop.FR.TP} nor \eqref{eq.prop.FR.IP} hold, then both power constraints are active: $\mu_1, \mu_2 >0$; closed-form solutions for $\mu_1, \mu_2$ in this case are not known and they have to be found numerically from \eqref{eq.thm.R*.2}. Note that $\bR^*$ in \eqref{eq.R*.general} generalizes the respective WF full-rank solution $\bR_{WF} = \mu_1^{-1}\bI - \bW_1^{-1}$ (no IPC) to the interference-constrained environment, where $\bW_{\mu}^{-2}$ takes the role of $\mu_1^{-1}\bI$.

While the optimal covariance in \eqref{eq.prop.FR.TP.2} is the same as the standard WF over the eigenmodes of $\bW_1$ (subject to the power constraint only), its range of validity is different: while the standard WF solution is of full rank under only the lower bound in \eqref{eq.prop.FR.TP} (this can be obtained by setting $P_I=\infty$ so that only the lower bound remains), i.e. the WF optimal covariance is of full rank for \textit{all} sufficiently high power/SNR, the interference constraint also imposes the upper bound in \eqref{eq.prop.FR.TP} and hence its full-rank solution will not hold if $P_T$ is too high, a remarkable difference to the standard WF.

Next, we explore the case where $\bW_2$ is of rank 1. This models the case when a primary user has a single-antenna receiver or when its channel is a keyhole channel, see e.g. \cite{Chizhik-02}\cite{Levin-08}.

\begin{prop}
\label{prop.W2.r1}
Let $\bW_1$ be of full rank and $\bW_2$ be of rank-1, so that $\bW_2=\gl_2\bu_2\bu_2^+$, where $\gl_2>0$ and $\bu_2$ are its active eigenvalue and eigenvector. If
\bal
\label{eq.prop.W2-r1.1}
&P_I \ge P_{I,th} =  m^{-1}\gl_2(P_T + tr(\bW^{-1}_1)) -\gl_2\bu_2^+\bW^{-1}_1\bu_2\\
\label{eq.prop.W2-r1.1a}
&P_T> m\gl_1(\bW^{-1}_1) -tr(\bW^{-1}_1)
\eal
then the IPC is redundant, the optimal covariance is of full rank and is given by the standard WF solution,
\bal
\label{eq.prop.W2-r1.R*1}
\bR^* = \bR^*_{WF} = \mu_{WF}^{-1}\bI - \bW^{-1}_1
\eal
where $\mu_{WF}^{-1} = m^{-1}(P_T+tr(\bW^{-1}_1))$.

If
\bal
\label{eq.prop.W2-r1.3}
&\gl_2\gl_1(\bW^{-1}_1)-\gl_2\bu_2^+\bW^{-1}_1\bu_2 < P_I < P_{I,th},\\
\label{eq.prop.W2-r1.4}
 & P_T > m\gl_2^{-1}P_I + m\bu_2^+\bW^{-1}_1\bu_2 -tr(\bW^{-1}_1)
\eal
then the IPC and TPC are active, the optimal covariance is of full rank and is given by
\bal
\label{eq.prop.W2-r1.R*2}
\bR^* = \mu_1^{-1}\bI - \bW^{-1}_1 -\alpha\bu_2\bu_2^+
\eal
where $\alpha=\mu_1^{-1}-(\mu_1+\gl_2\mu_2)^{-1}$, and $\mu_1,\ \mu_2>0$ are found from \eqref{eq.thm.R*.2}:
\bal
\label{eq.prop.W2-r1.5}
\begin{aligned}
\mu_1 &= (P_T-\gl_2^{-1} P_I -\bu_2^+\bW^{-1}_1\bu_2+tr(\bW^{-1}_1))^{-1} (m-1)\\
\mu_2 &= (P_I +\gl_2\bu_2^+\bW^{-1}_1\bu_2)^{-1}-\gl_2^{-1}\mu_1
\end{aligned}
\eal
\end{prop}
\begin{proof}
See Appendix.
\end{proof}

Note that the 1st two terms in \eqref{eq.prop.W2-r1.R*2} represent the standard WF solution while the last term is a correction due to the IPC, which is reminiscent of a partial null forming in an adaptive antenna array, see e.g. \cite{VanTrees-02}.

\section{Inactive Constrains}

It is straightforward to see both constraints cannot be inactive at the same time (since the capacity is a strictly increasing function of the Tx power without the IPC). In this section, we explore the scenarios when one of the two constraints is redundant\footnote{"inactive" implies "redundant" but the converse is not true: for example, inactive TPC means $tr\bR^* < P_T$ and this implies $\mu_1=0$ (from complementary slackness) so that it is also redundant (can be omitted without affecting the capacity), but $\mu_1=0$ does not imply  $tr\bR^* < P_T$ since $tr\bR^* = P_T$ is also possible in some cases.}.

Note that, unlike the standard WF where the TPC is always active, it can be inactive under the IPC, which is ultimately due to the interplay of interference and power constraints. The following proposition explores this in some details.

\begin{prop}
\label{prop.TPC.inact}
The TPC is redundant only if $\mathcal{N}(\bW_2) \in \mathcal{N}(\bW_1)$ and is active otherwise. In particular, it is active (for any $P_T$ and $P_I$) if $r(\bW_1)> r(\bW_2)$, e.g. if $\bW_1$ is full-rank and $\bW_2$ is rank-deficient.
\end{prop}
\begin{proof}
Use \eqref{eq.thm.Inac.TPC.3} and note that this condition is necessary for the TPC to be redundant (since the KKT conditions are necessary for optimality and $\mu_1=0$ is also necessary for the TPC to be redundant). Now, if $r(\bW_1)> r(\bW_2)$, then
\bal
\dim(\mathcal{N}(\bW_2)) =m-r(\bW_2) > m-r(\bW_1) =\dim(\mathcal{N}(\bW_1))
\eal
where $\dim(\mathcal{N})$ is the dimensionality of $\mathcal{N}$, and hence $\mathcal{N}(\bW_2) \in \mathcal{N}(\bW_1)$ is impossible so that the TPC is active.
\end{proof}

If $\bW_2$ is full-rank, a more specific result can be established.

\begin{prop}
\label{prop.IPC}
If $\bW_2>0$, then the TPC is redundant if and only if $P_T \ge tr(\bR^*)$, where the optimal covariance $\bR^*$ is as follows
\bal
\label{eq.R*.IPC}
\bR^* = \bW_2^{-\frac{1}{2}}(\mu_2^{-1}\bI- \bW_2^{\frac{1}{2}}\bW_1^{-1}\bW_2^{\frac{1}{2}})_+\bW_2^{-\frac{1}{2}},
\eal
and $\mu_2>0$ is
\bal
\label{eq.mu2.IPC}
\mu_2^{-1}= \frac{1}{r_+}P_I+ \frac{1}{r_+}\sum_{i=1}^{r_+} \gl_{bi}^{-1},
\eal
where $\gl_{bi}=\gl_i(\bW_2^{-1}\bW_1)$ and $r_+=r(\bR^*)$ is the rank of the optimal covariance matrix, which can be found as the largest solution of the following inequality:
\bal
\label{eq.PI.IPC}
P_I> \sum_{i=1}^{r_+} (\gl_{br_+}^{-1}- \gl_{bi}^{-1})_+
\eal

In particular, this holds if
\bal
\label{eq.PT.IPC}
P_T \ge \gl_m^{-1}(\bW_2)P_I
\eal
\end{prop}
\begin{proof}
Let $\bR^*$ be as in \eqref{eq.R*.IPC}. Then, $C\le C(\bR^*)$, since $C(\bR^*)$ is the capacity under the IPC only (this follows since \eqref{eq.R*.IPC} is a special case of \eqref{eq.thm.R*.1} with $\mu_1=0$). On the other hand, if $P_T \ge tr(\bR^*)$, then $\bR^*$ is feasible under the joint (IPC+TPC) constraints and hence $C\ge C(\bR^*)$, which proves the equality $C=C(\bR^*)$ and hence $\bR^*$ is optimal. \eqref{eq.mu2.IPC} follows from $tr(\bW_2\bR^*)=P_I$:
\bal
P_I =tr(\bW_2\bR^*) = tr(\mu_2^{-1}\bI- \bW_2^{\frac{1}{2}}\bW_1^{-1}\bW_2^{\frac{1}{2}})_+ = \sum_{i=1}^{r_+} (\mu_2^{-1}- \gl_{bi}^{-1})
\eal
\eqref{eq.PI.IPC} ensures, from \eqref{eq.R*.IPC}, that the rank of $\bR^*$ is $r_+$. Further note that \eqref{eq.PT.IPC} implies $P_T \ge tr(\bR^*)$:
\bal
P_T \ge \gl_m^{-1}(\bW_2)P_I
    =\gl_m^{-1}(\bW_2)tr(\bW_2\bR^*) \ge tr(\bR^*)
\eal
since $\bW_2\ge \gl_m(\bW_2)\bI$, but the converse is not true.
\end{proof}

Note that Proposition \ref{prop.IPC} gives a closed-form solution (including dual variables) for an optimal signaling strategy in the interference-limited regime (when the TPC is redundant and hence the IPC is active). The following proposition gives a sufficient and necessary condition for the IPC to be redundant (and hence the TPC is automatically active).

\begin{prop}
\label{prop.IPC.inact}
The IPC is redundant and hence the standard WF solution is optimal, $\bR^*=\bR_{WF}$, if and only if
\bal
\label{eq.prop.IPC.inact.1}
P_I \ge tr(\bW_2\bR_{WF})
\eal
in which case the TPC is active: $tr(\bR^*)=P_T$. In particular, this holds if
\bal
\label{eq.prop.IPC.inact.2}
P_I \ge \gl_1(\bW_2)P_T
\eal
\end{prop}
\begin{proof}
Note that, under the joint (IPC+TPC) constraint, $C\le C(\bR_{WF})$ (since $C(\bR_{WF})$ is attained by relaxing the IPC and retaining the TPC only, which cannot decrease the optimum). On the other hand, under the stated conditions, $\bR_{WF}$ is feasible under the joint constraints and hence the upper bound is achieved, $C=C(\bR_{WF})$ and $\bR_{WF}$ is optimal. It is straightforward to see that \eqref{eq.prop.IPC.inact.2} implies \eqref{eq.prop.IPC.inact.1}, since $\bW_2\le \gl_1(\bW_2)\bI$ (note that \eqref{eq.prop.IPC.inact.2} is sufficient but not necessary).
\end{proof}

While Proposition \ref{prop.IPC.inact} gives sufficient and necessary conditions for the IPC to be redundant, they depend on $P_T$ and $P_I$. The next proposition gives a sufficient condition which is independent of $P_T$ and $P_I$.

\begin{prop}
Let $\mathcal{R}(\bW_2) \in \mathcal{N}(\bW_1)$. Then, the IPC is redundant for any $P_I$ and $P_T$ and the standard WF solution is optimal: $\bR^*=\bR_{WF}$.
\end{prop}
\begin{proof}
As in Proposition \ref{prop.IPC.inact}, $C\le C(\bR_{WF})$. It is straightforward to verify that any active eigenvector of $\bR_{WF}$ is orthogonal to $\mathcal{N}(\bW_1)$ and hence, due to $\mathcal{R}(\bW_2) \in \mathcal{N}(\bW_1)$, $\bW_2\bR_{WF}=0$, i.e. the IPC is redundant and  $\bR_{WF}$ is feasible (for any $P_T$ and $P_I$) and hence $C(\bR_{WF})\le C$, from which the desired result follows.
\end{proof}

It follows from this Proposition that ZF transmission is optimal in this case, for any $P_I$ and $P_T$, so that there is no loss of optimality in using this popular signaling technique.

\section{Rank-1 Solutions}

In this section, we explore the case when $\bW_1$ is rank-one. As we show below, beamforming is optimal in this case. A practical appeal of this is due to its low-complexity implementation. Furthermore, rank-one $\bW_1$ is also motivated by single-antenna mobile units while the base station is equipped with multiple antennas, or when the MIMO propagation channel is of degenerate nature resulting in a keyhole effect, see e.g. \cite{Chizhik-02}\cite{Levin-08}.

We begin with the following result which bounds the rank of optimal covariance in any case.

\begin{prop}
\label{prop.rR}
If the TPC is active or/and $\bW_2$ is full-rank, then the rank of the optimal covariance $\bR^*$ of the problem (P1) in \eqref{eq.C.def} under the constraints in \eqref{eq.SR.2} is bounded as follows:
\bal
\label{eq.prop.rR}
r(\bR^*) \le r(\bW_1)
\eal
If the TPC is redundant and $\bW_2$ is rank-deficient, then there exists an optimal covariance $\bR^*$ of (P1) under the constraints in \eqref{eq.SR.2} that also satisfies this inequality\footnote{Under these conditions, an optimal covariance has an unusual property of being not necessarily unique, see an example in Section \ref{sec.Numerical Experiments}, so that there may exist extra solutions that do not salsify this inequality, but all of them deliver the same capacity.}.
\end{prop}
\begin{proof}
See Appendix.
\end{proof}

The following are immediate consequences of Proposition \ref{prop.rR}.

\begin{cor}
If $\bW_2$ is of full-rank or/and if the TPC is active, then the optimal covariance $\bR^*$ is of full-rank only if $\bW_1$ is of full-rank (i.e. rank-deficient $\bW_1$ ensures that $\bR^*$ is also rank-deficient).
\end{cor}

\begin{cor}
If $r(\bW_1)=1$, then $r(\bR^*)=1$, i.e. beamforming is optimal.
\end{cor}

Note that this rank (beamforming) property mimics the respective property for the standard WF. However, while signalling on the (only) active eigenvector of $\bW_1$ is optimal under the standard WF (no IPC), it is not so when the IPC is active, as the following result shows. To this end, let $\bW_1=\gl_1\bu_1\bu_1^+$, i.e. it is rank-1 with $\gl_1>0, \bu_1$ be the (only) active eigenvalue and eigenvector; $\gamma_I=P_I/P_T$ be the "interference-to-signal" ratio, and
\bal
\gamma_1 = \frac{\bu_1^+\bW_2^\dag\bu_1}{\bu_1^+(\bW_2^\dag)^2\bu_1},\ \gamma_2 = \bu_1^+\bW_2\bu_1
\eal
where $\bW_2^\dag$ is Moore-Penrose pseudo-inverse of $\bW_2$; $\bW_2^\dag=\bW_2^{-1}$ if $\bW_2$ is full-rank \cite{Horn-85}.

\begin{prop}
\label{prop.r1.IPC}
Let $\bW_1$ be rank-1.

1. If $\gamma_I < \gamma_1$, then the TPC is redundant and the optimal covariance can be expressed as follows
\bal
\label{eq.prop.r1.IPC.1}
\bR^* = P_I\frac{\bW_2^\dag\bu_1\bu_1^+\bW_2^\dag} {\bu_1^+\bW_2^\dag\bu_1}
\eal
The capacity is
\bal
\label{eq.prop.r1.IPC.2}
C = \log(1+\gl_1 \alpha P_T)
\eal
where $\alpha = \gamma_I \bu_1^+\bW_2^\dag\bu_1 < 1$.

2. If $\gamma_I\ge \gamma_2$, then the IPC is redundant and the standard WF solution applies: $\bR^* = P_T\bu_1\bu_1^+$. This condition is also necessary for the optimality of $P_T\bu_1\bu_1^+$ under the TPC and IPC when $\bW_1$ is rank-1. The capacity is as in \eqref{eq.prop.r1.IPC.2} with $\alpha=1$.

3. If $\gamma_1 \le \gamma_I < \gamma_2$, then both constraints are active. The optimal covariance is
\bal
\label{eq.prop.r1.IPC.3}
\bR^* = P_T\frac{\bW_{2\mu}^{-1}\bu_1\bu_1^+\bW_{2\mu}^{-1}} {\bu_1^+\bW_{2\mu}^{-2}\bu_1}
\eal
where $\bW_{2\mu}=\bI+\mu_2\bW_2$, and $\mu_2>0$ is found from the IPC: $tr(\bW_2\bR^*)=P_I$. The capacity is as in \eqref{eq.prop.r1.IPC.2} with
\bal
\alpha= (\bu_1^+\bW_{2\mu}^{-1}\bu_1)^2 |\bW_{2\mu}^{-1}\bu_1|^{-2} \le 1
\eal
with equality if and only if $\bu_1$ is an eigenvector of $\bW_2$.
\end{prop}
\begin{proof}
See Appendix.
\end{proof}

Note that the optimal signalling in case 1 is along the direction of $\bW_{2}^{\dag}\bu_1$ and not that of $\bu_1$ (unless $\bu_1$ is also an eigenvector of $\bW_{2}$), as would be the case for the standard WF with redundant IPC. In fact, $\bW_{2}^{\dag}$ plays a role of a "whitening" filter here. Similar observation applies to case 3, with $\bW_{2}$ replaced by $\bW_{2\mu}$. $\alpha$ in Proposition \ref{prop.r1.IPC} quantifies power loss due to enforcing the IPC; $\alpha=1$ means no power loss.

\section{Identical Eigenvectors}
\label{sec.Iden.EV}

In this section, we consider a scenario where $\bW_1$ and $\bW_2$ have the same eigenvectors. This may be the case when the scattering environment around the base station (the Tx) is the same as seen from the Rx and primary user U. In this case, the general solution in Theorem 1 significantly simplifies to the following:
\bal
\label{eq.thm.R*D}
\bR^* = \bU\bLam^*\bU^+
\eal
where unitary matrix $\bU$ collects eigenvectors of $\bW_{1(2)}$ and diagonal matrix $\bLam^*=\diag\{\gl_i^*\}$ collects the eigenvalues of $\bR^*$:
\bal
\label{eq.thm.R*D.2}
\gl_i^* = [(\mu_1+\gl_{2i}\mu_2)^\dag - \gl_{1i}]_+
\eal
where $\gl_{1i(2i)}$ are the eigenvalues of $\bW_{1(2)}$. Dual variables $\mu_{1(2)} \ge 0$ are determined from the following:
\bal
\label{eq.thm.R*D.3}
\mu_1\left(\sum_i \gl_i^* - P_T\right)=0,\ \mu_2\left(\sum_i \gl_{2i}\gl_i^* -P_I\right)=0
\eal
subject to $\sum_i \gl_i^* \le P_T$, $\sum_i \gl_{2i}\gl_i^* \le P_I$. The capacity can be expressed as in \eqref{eq.C} with $\gl_{ai} = \gl_{1i}(\mu_1+ \mu_2\gl_{2i})^\dag$.

Note that, in this case, signaling on the eigenmodes of the main channel $\bW_1$ is optimal, but power allocation is not given by the standard WF, unless the IPC is redundant ($\mu_2=0$). It follows from \eqref{eq.thm.R*D.2} that $\gl_i^*=0$ (power allocated to $i$-th eigenmode is zero) if $\mu_1=0$ (redundant TPC) and $\gl_{2i}=0$ ($i$-th eigenmode of 2nd channel $\bW_2$ is inactive), in addition to the standard WF property that $\gl_i^*=0$ if $\gl_{1i}=0$ under the active TPC ($\mu_1>0$).

In the context of massive MIMO systems under favorable propagation, see e.g. \cite{Marzetta-16}, $\bW_1$ and $\bW_2$ become diagonal matrices (and thus have the same eigenvectors), and so is $\bR^*$, i.e. independent signaling is optimal, and the solution in \eqref{eq.thm.R*D.2} gives the optimal power allocation in such setting. This significantly simplifies its implementation since numerical complexity of generic convex solvers can be prohibitively large for massive MIMO settings.

\section{Iterative Bisection Algorithm}
\label{sec.IBA}

While Theorem \ref{thm.R*} gives a closed-form solution for an optimal covariance $\bR^*$ up to dual variables, no closed-form solution is known for \eqref{eq.thm.R*.2} in the general case; the sections above provided complete closed-form solutions in some special cases. In this section, we develop an iterative numerical algorithm to solve \eqref{eq.thm.R*.2} in the general case in an efficient way and prove its convergence.

First, we consider the standard bisection algorithm \cite{Boyd-04}. Let $f(x)$ be a function with the following property: $f(x)\ge 0$ for any $x< x_0$ and $f(x)\le0$ for any $x> x_0$, where $x_0$ is a solution of $f(x)=0$. Then, the following bisection algorithm (BA) can be used to solve $f(x)=0$, where $x_l, x_u$ are the upper and lower bounds to $x_0$: $x_l\le x_0 \le x_u$, and $\epsilon >0$ is any desired accuracy. In fact, it is straightforward to show that this algorithm will converge in a finite number $N$ of steps such that
\bal
N \le \left\lceil\log_2\left(\frac{x_u-x_l}{\epsilon}\right)\right\rceil
\eal
where $\lceil\cdot\rceil$ denotes ceiling, so that the convergence is exponentially fast and hence the algorithm is very efficient \cite{Boyd-04}.

\begin{algorithm}
\caption{Bisection algorithm (BA)}
\label{alg.BA}
\begin{algorithmic}
\Require $f(x),\ x_l,\ x_u,\ \epsilon$
\Repeat
\State 1. Set $x=\frac{1}{2}(x_l+x_u)$.
\State 2. If $f(x)<0$, set $x_u=x$. Otherwise, set $x_l=x$. Terminate if $f(x)=0$.
\Until $|x_u-x_l| \le \epsilon$.
\end{algorithmic}
\end{algorithm}

An alternative stopping criteria for this algorithm is $|f(x)| \le \epsilon$ and the two criteria are equivalent when $f(x)$ is continuous.

The BA can be used to solve for $\mu_1$, $\mu_2$ in \eqref{eq.thm.R*.2} in an iterative way, as we show below. To this end, we need to establish lower and upper bounds to the solutions $\mu_1^*,\ \mu_2^*$ required by the BA.

\begin{prop}
\label{prop.mu.bounds}
Let $\mu_1^*,\ \mu_2^*$ be solutions of \eqref{eq.thm.R*.2}, i.e. the optimal dual variables. They can be bounded as follows:
\bal
\label{eq.mu.l.u}
&0\le \mu_1^* \le \mu_{1u} =m(P_T+\gl_1^{-1}(\bW_1))^{-1}\\
\label{eq.mu2.l.u}
&0\le \mu_2^* \le \mu_{2u} =(P_I/r_2+\gl_m(\bW_2)/\gl_1(\bW_1))^{-1}
\eal
where $r_2$ is the rank of $\bW_2$ and $m$ is the number of Tx antennas.
\end{prop}
\begin{proof}
From the KKT conditions in \eqref{eq.T.KKT},
\bal
\label{eq.prop.muu.4}
(\bI+\bW_1\bR)^{-1}\bW_1\bR = \mu_1\bR+ \mu_2\bW_2\bR,\ \
\mu_1 P_T+ \mu_2 P_I = tr((\bI+\bW_1\bR)^{-1}\bW_1\bR)
\eal
Let $\bA$  be a matrix with positive eigenvalues, $\gl_i(\bA)\ge 0$. Since
\bal
\gl_i((\bI+\bA)^{-1}\bA) = \gl_i(\bA)(1+\gl_i(\bA))^{-1} \le \gl_1(\bA)(1+\gl_1(\bA))^{-1}
\eal
it follows that
\bal
tr((\bI+\bA)^{-1}\bA) \le m\gl_1(\bA)(1+\gl_1(\bA))^{-1}
\eal
Now use $\bA=\bW_1\bR$ to obtain
\bal
tr((\bI+\bW_1\bR)^{-1}\bW_1\bR) \le m\gl_1(\bW_1\bR)(1+\gl_1(\bW_1\bR))^{-1} \le m P_T(\gl_{11}^{-1}+P_T)^{-1}
\eal
where $\gl_{11}=\gl_1(\bW_1)$ and $\gl_1(\bW_1\bR) \le \gl_{11}P_T$, so that 2nd inequality in \eqref{eq.mu.l.u} follows from \eqref{eq.prop.muu.4}. Let $\gl_{2m}=\gl_m(\bW_2)$. Using \eqref{eq.thm.R*.1} under active IPC,
\bal
P_I= tr(\bW_2\bR^*) \le tr(\bW_2(\bW_{\mu}^{\dag})^2) (1- \gl_{11}^{-1}(\mu_1+\mu_2\gl_{2m}))
\le r_2 (\mu_2^{-1}-\gl_{2m}\gl_{11}^{-1})
\eal
from which 2nd inequality in \eqref{eq.mu2.l.u} follows, where we have used \bal
tr(\bW_2(\bW_{\mu}^{\dag})^2) \le \mu_2^{-1} tr(\bW_2\bW_2^{\dag})= \mu_2^{-1}r_2
\eal
and $\bW \ge \gl_m(\bW)\bI$ for any Hermitian $\bW$. If the IPC is inactive, $\mu_2=0$ and the inequality holds in obvious way.
\end{proof}

To proceed further, let
\bal
x_{\epsilon} = \mathrm{BA}[f(x),x_l,x_u,\epsilon]
\eal
formally denote an $\epsilon$-accurate solution of $f(x)=0$ given by the BA and let
\bal
f_1(\mu_1,\mu_2)= \mu_1(tr(\bR^*(\mu_1,\mu_2)) - P_T),\
f_2(\mu_1,\mu_2)= \mu_2(tr(\bW_2\bR^*(\mu_1,\mu_2))-P_I)
\eal
where $\bR^*(\mu_1,\mu_2)$ denotes $\bR^*$ in \eqref{eq.thm.R*.1} for \textit{given} $\mu_1,\ \mu_2$. Then, the optimal dual variables $\mu_1^*,\mu_2^*$ satisfy $f_1(\mu_1^*,\mu_2^*)=0$ and $f_2(\mu_1^*,\mu_2^*)=0$. For a given $\mu_2^*$, one could use the BA to formally express $\mu_1^*$ as
\bal
\mu_1^* = \mathrm{BA}[f(x)=f_1(x,\mu_2^*),\mu_l,\mu_{1u},0]
\eal
where, from \eqref{eq.mu.l.u}, $\mu_l=0$, and likewise for $\mu_2^*$ (since the convergence of the BA is exponentially fast, the inaccuracy $\epsilon$ can be set to be arbitrary small in practice so that we disregard here this small inaccuracy by setting $\epsilon=0$ to simplify the analysis; numerical experiments support this approach). The following proposition shows that $f_1(x,\mu_2),\ f_2(\mu_1,x)$ have the property needed for the convergence of the BA as stated above. To this end, let $P_1(\mu_1,\mu_2)=tr(\bR^*(\mu_1,\mu_2))$, $P_2(\mu_1,\mu_2)=tr(\bW_2\bR^*(\mu_1,\mu_2))$, i.e. the transmit and interference powers for given $\mu_1,\mu_2$.

\begin{prop}
\label{prop.f12.prop}
Let $\mu_{10}$ be a solution of $f_1(x,\mu_2)= 0$ for a given $\mu_2$ and subject to $P_1(x,\mu_2)\le P_T$. Then, $f_1(\mu,\mu_2)\ge 0$ for any $\mu< \mu_{10}$ and $f_1(\mu_1,\mu_2)\le 0$ for any $\mu_1> \mu_{10}$. Likewise, if $\mu_{20}$ is a solution of $f_2(\mu_1,x)= 0$ for a given $\mu_1$ and subject to $P_2(\mu_1,x)\le P_I$, then $f_2(\mu_1,\mu_2)\ge 0$ for any $\mu_2< \mu_{20}$ and $f_2(\mu_1,\mu_2)\le 0$ for any $\mu_2> \mu_{20}$.
\end{prop}
\begin{proof}
See Appendix.
\end{proof}

Thus, this proposition shows that the BA can be used to solve $f_1(x,\mu_2)=0$ for a given $\mu_2$ and likewise for $f_2(\mu_1,x)=0$.
Unfortunately, neither of the optimal dual variables is known in advance. Hence, we propose the following iterative bisection algorithm (IBA) which finds optimal dual variables without such advance knowledge.

\begin{algorithm}
\caption{Iterative Bisection Algorithm (IBA)}
\label{alg.IBA}
\begin{algorithmic}
\Require $f_1(\mu_1,\mu_2),\ f_2(\mu_1,\mu_2),\ \mu_{1u},\ \mu_{2u},\ \delta$
\State 1. Set $\mu_{20}=0$, $k=1$.
\Repeat
\State 2. Set $\mu_{1k}= \mathrm{BA}[f_1(x,\mu_{2(k-1)}),0,\mu_{1u},\delta]$.
\State 3. Set $\mu_{2k}= \mathrm{BA}[f_2(\mu_{1k},x),0,\mu_{2u},\delta]$.
\State 4. $k:=k+1$.
\Until stopping criterion is met.
\end{algorithmic}
\end{algorithm}

Note that the BA used in steps 2 and 3 will converge, as follows from Proposition \ref{prop.f12.prop}. A possible stopping criteria for this algorithm is $|f_{1(2)}(\mu_{1k},\mu_{2k})| \le \epsilon$ or when a number of steps exceeds maximum $k_{max}$. The following proposition shows that the IBA generates converging sequences of dual variables $\{\mu_{1k}\}$, $\{\mu_{2k}\}$ under a mild technical condition.

\begin{prop}
\label{prop.IBA.conv}
The sequences $\{\mu_{1k}\}_{k=1}^{\infty}$, $\{\mu_{2k}\}_{k=1}^{\infty}$ generated by the IBA above converge if $\delta=0$ and $P_{1(2)}(\mu_{1},\mu_{2})$ are decreasing functions of $\mu_{1}, \mu_{2}$. In particular, this holds in any of the following cases:

1. The IPC is redundant, in which case the IBA converges in 1 iteration.

2. $\bW_1$ and $\bW_2$ have the same eigenvectors.

3. $\bR^*(\mu_1,\mu_2)$ is full-rank.
\end{prop}
\begin{proof}
See Appendix.
\end{proof}

The following proposition shows that any stationary (and hence convergence) point of the IBA solves the dual optimality conditions in \eqref{eq.thm.R*.2}.

\begin{prop}
Any stationary point of the IBA is a solution of \eqref{eq.thm.R*.2} if $\delta=0$. Hence, the IBA converges to a solution of \eqref{eq.thm.R*.2} under the conditions of Proposition \ref{prop.IBA.conv}.
\end{prop}
\begin{proof}
Let $\mu_{1s},\ \mu_{2s}$ be a stationary point of the IBA, so that
\bal
\mu_{1s}= \mathrm{BA}[f_1(x,\mu_{2s}),0,\mu_{1u},0],\
\mu_{2s}= \mathrm{BA}[f_2(\mu_{1s},x),0,\mu_{2u},0]
\eal
It follows from 1st equality that $f_1(\mu_{1s},\mu_{2s})=0$ and $f_2(\mu_{1s},\mu_{2s})=0$ from 2nd one. Thus, $\mu_{1s},\ \mu_{2s}$ solves \eqref{eq.thm.R*.2}. Since a convergence point is stationary, it follows that the IBA converges to a solution of \eqref{eq.thm.R*.2}.
\end{proof}

While the analytical convergence results above are limited to $\delta=0$, $\delta>0$ is used in practice. Since the BA converges exponentially fast, very small $\delta$ can be selected in the IBA without significant increase in computational complexity of each step and hence the analysis serves as a reasonable approximation (due to the continuity of the problem and functions involved). Furthermore, numerous numerical experiments indicate that the IBA always converges, even when the conditions 1-3 of Proposition \ref{prop.IBA.conv} are not met (we were not able to observe a single case where it did not). In the majority of the studied cases, a small to moderate number of IBA iterations (1...50) is needed to achieve a high accuracy of $10^{-5}$, while up to 250 iterations are required in some exceptional cases with $\epsilon=10^{-10}$ (which is hardly required in practice).

\section{Extension to Multi-User Environments}
\label{sec.multi}

In a typical multi-user environment, there are multiple users to which interference has to be limited, so that the problem in \eqref{eq.C.def} is solved under the following constraint set:
\bal
\label{eq.SR.3}
S_R=\{\bR: \bR\ge 0,\ tr(\bR) \le P_T,\ tr(\bW_{2k}\bR) \le P_{Ik},\ k=1..K\},
\eal
where $\bW_{2k}=\bH_{2k}^+\bH_{2k}$ and $P_{Ik}$ represent channel to user $k$ and respective interference constraint power, $K$ is the number of users, see Fig. \ref{fig.muli-user}. Using the same approach as in Theorem \ref{thm.R*}, it is straightforward to see that Theorem \ref{thm.R*} applies with
\bal
\bW_{\mu}=(\mu_1\bI+ \sum_k \mu_{2k}\bW_{2k})^{\frac{1}{2}}
\eal
where dual variables are found from the following system of non-linear equations:
\bal
\label{eq.thm.R*.2m}
\mu_1(tr(\bR^*) - P_T)=0,\ \mu_{2k}(tr(\bW_{2k}\bR^*)-P_{Ik})=0,\ k=1..K
\eal
subject to $\mu_1,\mu_{2k} \ge 0,\ tr(\bR^*)\le P_T$, $tr(\bW_{2k}\bR^*) \le P_{Ik}$. In particular, the iterative bisection algorithm of Section \ref{sec.IBA} can be used with a proper extension to accommodate multiple users.

\begin{figure}[t]
	\centerline{\includegraphics[width=3.5in]{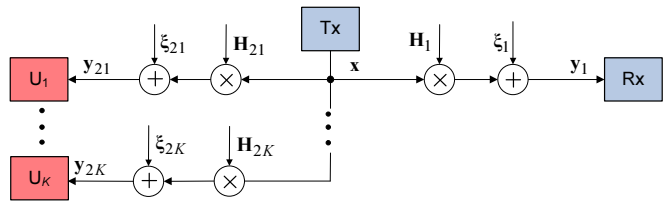}}
	\caption{A block diagram of multi-user Gaussian MIMO channel under interference constraints. $\bH_1$ and $\bH_{2k}$ are the channel matrices to the Rx and $k$-th user respectively. Interference constraints are to be satisfied for each user.}
		\label{fig.muli-user}
\end{figure}

One may also consider the total (rather than individual) interference power constraint so that
\bal
\label{eq.SR.4}
S_R=\{\bR: \bR\ge 0,\ tr(\bR) \le P_T,\ \sum_k tr(\bW_{2k}\bR) \le P_{I} \}
\eal
In this case, all the above results apply with the substitution
\bal
\bW_2 = \sum_k \bW_{2k}
\eal

\section{Examples}
\label{sec.Numerical Experiments}

In this section, we present some numerical results that illustrate the analytical results above as well as the performance of the IBA.

\textbf{Example 1:} In this example, $P_I=1$ and
\bal 
\label{eq.W.example.1}
\bW_1 =
\left[
\begin{array}{cc}
    1 & 0\\
    0 & 0.5\\
\end{array}
\right],\
\bW_2 =
\left[
\begin{array}{cc}
    1 & -0.5\\
   -0.5 & 1\\
\end{array}
\right]
\eal
Fig. 3 shows the number of iterations of the IBA required to solve \eqref{eq.thm.R*.2} with the accuracy $\epsilon=10^{-5}$ vs. $P_T$; the optimal dual variables $\mu_1^*,\ \mu_2^*$ as well as the actual Tx and interference powers ($P_1=tr(\bR^*(\mu_1^*,\mu_2^*))$ and  $P_2=tr(\bW_2\bR^*(\mu_1^*,\mu_2^*))$ respectively) are also shown. Note the transition from the Tx power-limited regime (inactive IPC) to the interference-limited regime (inactive TPC) as $P_T$ increases, which is visible when the respective dual variable sharply decreases to 0. In particular, the IPC is inactive when $P_T< 1.1$ and the TPC is inactive when $P_T>1.8$, while both constraints are active otherwise. As $P_T$ increases, the IPC becomes active at about $P_T\approx 1.1$, at which point the required number of iteration sharply increases from 1 to 36 and then to 82, gradually decreasing to a small number of 2. When the IPC is inactive, the number of iterations is 1, in agreement with Proposition \ref{prop.IBA.conv}. As this example demonstrates, anyone of the constraints can be inactive depending on the $P_T, P_I$ and channel matrices.

\begin{figure}[t]
\centerline{\includegraphics[width=3.3in]{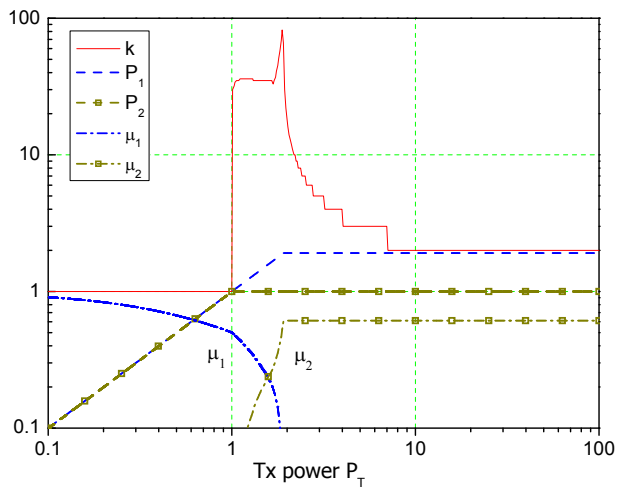}}
\caption{Convergence of the IBA, i.e. the number $k$ of iterations required to achieve $\epsilon=10^{-5}$ vs. $P_T$; $\bW_1$ and $\bW_2$ are as in \eqref{eq.W.example.1}, $P_I=1$. $P_1, P_2, \mu_1^*, \mu_2^*$ are also shown.}
\label{fig.1}
\end{figure}

Fig. \ref{fig.C} shows the capacity under the joint (TPC+IPC) constraints for the channel of Fig. \ref{fig.1}, along with the capacities under the TPC ($C_{WF}$), which is given by the WF procedure, and the IPC ($C_{IPC}$) alone. Note that $C$ is upper bounded in general by $C_{WF}$ and $C_{IPC}$,
\bal
\label{eq.C.bound}
C \le \min\{C_{WF},C_{IPC}\}
\eal
and that this bound is tight: if the IPC is inactive (power-limited regime), then $C=C_{WF}$, and if the TPC is inactive (interference-limited regime), then $C=C_{IPC}$, so that the inequality is strict only in a (small) transition region (when both constraints are active simultaneously) and hence the following approximation can be used over the entire range of $P_T$:
\bal
\label{eq.C.approx}
C \approx \min\{C_{WF},C_{IPC}\}
\eal
Note also that the capacity does not grow unbounded, in agreement with Proposition \ref{prop.C.inf} (since $\bW_2$ is full-rank, so that $\sN(\bW_2)$ is empty and hence $\sN(\bW_2) \in \sN(\bW_1)$). This changes significantly if $\bW_2$ is rank-deficient: the TPC is always active and the capacity grows unbounded, as the next example shows.

\begin{figure}[t]
\centerline{\includegraphics[width=3.3in]{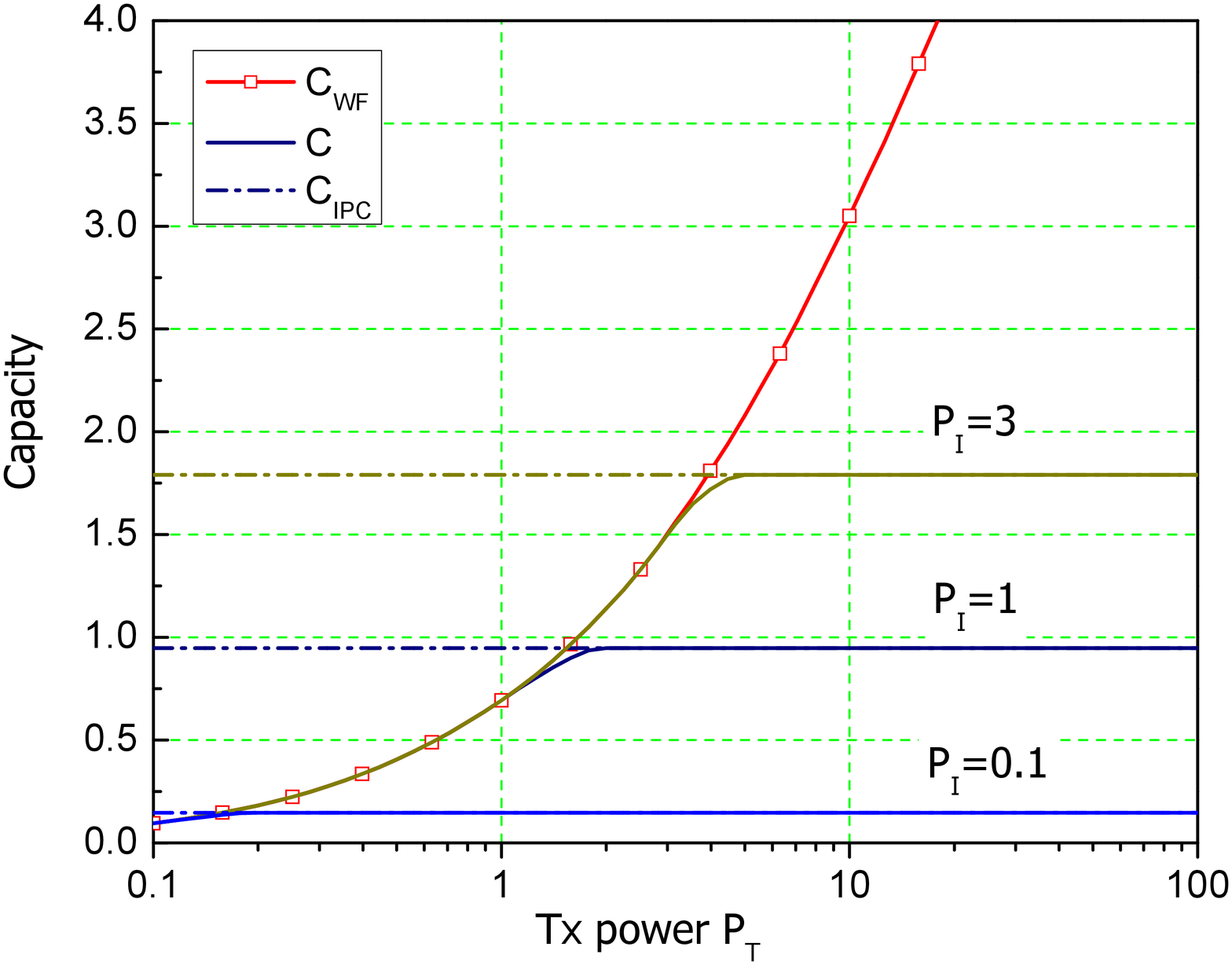}}
\caption{The capacity under the TPC ($C_{WF}$), IPC ($C_{IPC}$) and joint TPC+IPC ($C$) constraints for the same setting as in Fig. \ref{fig.1}; $P_I=0.1, 1$ and 3.}.
\label{fig.C}
\end{figure}

\begin{figure}[t]
\centerline{\includegraphics[width=3.3in]{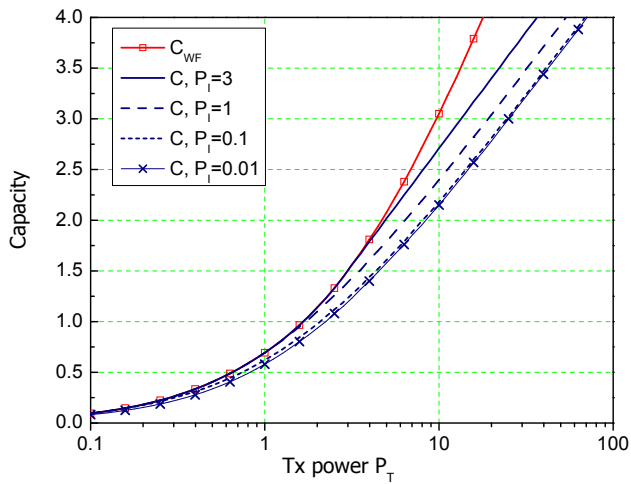}}
\caption{The capacity under the TPC ($C_{WF}$) and and the joint TPC+IPC ($C$) constraints for $\bW_1$ as in \eqref{eq.W.example.1} and $\bW_2$ as in \eqref{eq.W.example.2}; $P_I=0.01, 0.1, 1$ and 3.}.
\label{fig.C2}
\end{figure}

\textbf{Example 2:} In this example
\bal 
\label{eq.W.example.2}
\bW_2 =
\left[
\begin{array}{cc}
    1 & -1\\
   -1 & 1\\
\end{array}
\right]
\eal
so that $\bW_2$ is rank-deficient, and $\bW_1$ is as in Example 1. As Fig. \ref{fig.C2} shows, the capacity grows unbounded for any $P_I$, even small one, in agreement Corollary \ref{cor.C.inf} (since $r(\bW_1) > r(\bW_2)$) or Proposition \ref{prop.C.inf} (since $\sN(\bW_2) \notin \sN(\bW_1)$). In this case, the bound in \eqref{eq.C.bound} is tight and the approximation in \eqref{eq.C.approx} is valid at low SNR only, since $C_{IPC}=\infty$ (as $\bW_1$ is full-rank while $\bW_2$ is rank-deficient).

Comparing Fig. \ref{fig.C2} to Fig. \ref{fig.C}, one observes that while decreasing $P_I$ decreases the capacity in both cases, the behaviour is qualitatively different: the capacity saturates and variations in $P_I$ have major impact on the saturation level in Fig. \ref{fig.C} while the capacity grows unbounded in Fig. \ref{fig.C2} and variations in $P_I$ have moderate or small impact on its value (negligible if $P_I \le 0.1$). In both cases, the channel matrix of primary user has a major impact on the capacity at high power/SNR regime, while being negligible at low SNR. Its null space (or rank) determines the qualitative behaviour of the capacity at high power/SNR regime.

It should also be noted that the optimal covariance $\bR^*$ is not diagonal, even thoung $\bW_1$ is, when the IPC is active - a sharp distinction to the TPC constraint only, where $\bR^*$ and $\bW_1$ have the same eigenvectors so that diagonal $\bW_1$ implies diagonal $\bR^*$. Hence, introducing the IPC makes independent signaling sub-optimal for independent channels in general (unless $\bW_2$ is also diagonal or if the IPC is redundant).

\textbf{Example 3:} To demonstrate that $\bR^*$ is not necessarily unique under the IPC (a stark difference to the standard WF solution, where the optimal covariance is always unique, unless the capacity is zero), consider the following channel:
\bal 
\label{eq.W.example.3}
\bW_1 = \diag\{w_1,0\},\ \bW_2 = \diag\{w_2,0\},\ w_2 P_T > P_I
\eal
It is straightforward to see that an optimal covariance is
\bal 
\bR^* = \diag\{w_2^{-1}P_I, a\},\ 0 \le a \le P_T - w_2^{-1}P_I
\eal
so that it is not unique, but the capacity is:
\bal
C=\ln(1+ w_1 w_2^{-1}P_I)
\eal
for any $a$. While $r(\bR^*)=1$ for $a=0$ so that \eqref{eq.prop.rR} does hold, this is not the case for $a>0$, which is in stark contrast to the standard WF solution, where \eqref{eq.prop.rR} always holds. Note that this unusual property disappears if $w_2 P_T \le P_I$, in which case the IPC is redundant and the standard WF solution applies:
\bal
\bR^* = \diag\{P_T, 0\},\ C=\ln(1+w_1 P_T).
\eal

While the examples above are limited to $m=2$, the results are representative, i.e. similar tendencies also hold for larger $m$ and different channels.

\section{Conclusion}

Optimal signaling over the Gaussian MIMO channel is considered under interference constraints. The closed-form solution for an optimal Tx covariance is presented in the general case (up to dual variables). The iterative bisection algorithm is developed to evaluate numerically the dual variables in the general case and its convergence is proved for some special cases. A number of explicit closed-form solutions (without dual variables) are obtained, including full-rank, rank-1 (beamforming) and when the channels have the same eigenvectors. Sufficient and necessary conditions for the TPC and IPC being active/inactive (or redundant) are given and their interplay is investigated. It is pointed out that the TPC and IPC can be active simultaneously so that neither condition can be absorbed into another in general, as was sometimes suggested in the literature. Null space of the channel matrix of external (primary) user has a major impact on the capacity at high SNR while being negligible at low SNR. These analytical results can serve as building blocks for the analysis, design and optimization of multi-user MIMO networks in interference-limited environments, as in e.g. 5G scenarios with aggressive frequency re-use, HetNets and licensed/unlicensed usage to improve spectral efficiency.

\section{Appendix}

\vspace*{-.5\baselineskip}
\subsection{Proof of Theorem \ref{thm.R*}}

Since the problem is convex and Slater's condition holds, the KKT conditions are both sufficient and necessary for optimality \cite{Boyd-04}. They take the following form:
\bal
\label{eq.T.KKT}
&-(\bI+\bW_1\bR)^{-1}\bW_1-\bM+\mu_1\bI+ \mu_2\bW_2 = 0\\
\label{eq.T.KKT.2}
&\bM\bR=0,\ \mu_1(tr(\bR)-P_T)=0,\ \mu_2(tr(\bW_2\bR)-P_I)=0,\\
&\bM\ge 0,\ \mu_1\ge 0,\ \mu_2\ge 0\\
&tr(\bR) \le P_T,\ tr(\bW_2\bR) \le P_I,\ \bR\ge 0
\eal
where $\bM$ is Lagrange multiplier responsible for the positive semi-definite constraint $\bR\ge 0$. We consider first the case of full-rank $\bW_{\mu}$ (i.e. either $\mu_1>0$ or/and $\bW_2>0$), so that $\bW_{\mu}^{\dag}=\bW_{\mu}^{-1}$. Let us introduce new variables: $\tilde{\bR}= \bW_{\mu}\bR\bW_{\mu}$, $\tilde{\bW_1}= \bW_{\mu}^{-1}\bW_1\bW_{\mu}^{-1}$, $\tilde{\bM}= \bW_{\mu}^{-1}\bM\bW_{\mu}^{-1}$. It follows that $\tilde{\bM}\tilde{\bR}=0$ and \eqref{eq.T.KKT} can be transformed to
\bal
&(\bI+\tilde{\bW_1}\tilde{\bR})^{-1}\tilde{\bW_1} +\tilde{\bM} = \bI
\eal
for which the solution is
\bal
\label{eq.T.KKT.Rtil}
\tilde{\bR} = (\bI-\tilde{\bM})^{-1} - \tilde{\bW}_1^{-1} = (\bI- \tilde{\bW}_1^{-1})_+
\eal
(this can be established in the same way as for the standard WF in \eqref{eq.RWF}). Transforming back to the original variables results in \eqref{eq.thm.R*.1}. \eqref{eq.thm.R*.2} are complementary slackness conditions in \eqref{eq.T.KKT.2}; \eqref{eq.C} follows, after some manipulations, by using $\bR^*$ of \eqref{eq.thm.R*.1} in $C(\bR)$.

The case of singular $\bW_{\mu}$ is more involved. It implies $\mu_1=0$ so that $\bW_{\mu}=(\mu_2\bW_2)^{\frac{1}{2}}$. It follows from the KKT condition in \eqref{eq.T.KKT} that, for the redundant TPC ($\mu_1=0$),
\bal
\label{eq.thm.Inac.TPC.1}
\bQ_1(\bI+\bQ_1\bR\bQ_1)^{-1}\bQ_1+\bM=\mu_2\bW_2
\eal
where $\bQ_1=\bW_1^{1/2}$. Let $\bx\in \mathcal{N}(\bW_2)$, i.e. $\bW_2\bx=0$, then
\bal
\bx^+\bQ_1(\bI+\bQ_1\bR\bQ_1)^{-1}\bQ_1\bx +\bx^+\bM\bx=0
\eal
so that $\bx^+\bM\bx=0$ and $\bQ_1\bx=0$, since $\bM\ge0$ and $\bI+\bQ_1\bR\bQ_1 > 0$. Thus, $\mathcal{N}(\bW_2) \in \mathcal{N}(\bQ_1)=\mathcal{N}(\bW_1)$ and $\mathcal{N}(\bW_2) \in \mathcal{N}(\bM)$, i.e.
\bal
\label{eq.thm.Inac.TPC.3}
\mathcal{N}(\bW_2) \in \mathcal{N}(\bW_1) \cap \mathcal{N}(\bM)
\eal
and this condition is also necessary for the TPC to be redundant. Using \eqref{eq.thm.Inac.TPC.1}, \eqref{eq.thm.Inac.TPC.3} and introducing new variables
\bal
\bLam_2=\bU_2^+\bW_2\bU_2,\ \tilde{\bR}=\bU_2^+\bR\bU_2, \tilde{\bQ}_1=\bU_2^+\bQ_1\bU_2,\ \tilde{\bM}=\bU_2^+\bM\bU_2,
\eal
where $\bU_2$ is a unitary matrix of eigenvectors of $\bW_2$, one obtains
\bal
\bLam_2=\left(
          \begin{array}{cc}
            \bLam_{2+} & 0 \\
            0 & 0 \\
          \end{array}
        \right),
\tilde{\bQ}_1=\left(
          \begin{array}{cc}
            \bQ_{1+} & 0 \\
            0 & 0 \\
          \end{array}
        \right),
\tilde{\bM}=\left(
          \begin{array}{cc}
            \bM_{+} & 0 \\
            0 & 0 \\
          \end{array}
        \right),
\tilde{\bR}=\left(
          \begin{array}{cc}
            \bR_{+} & \bR_{12}\\
            \bR_{21} & \bR_{22} \\
          \end{array}
        \right)
\eal
where $\bLam_{2+}>0$ is a diagonal matrix of strictly positive eigenvalues of $\bW_2$, so that \eqref{eq.thm.Inac.TPC.1} can be transformed to
\bal
\label{eq.Inac.TPC.2}
\bQ_{1+}(\bI+\bQ_{1+}\bR_+\bQ_{1+})^{-1} \bQ_{1+} + \bM_+=\mu_2\bLam_{2+}>0
\eal
Using
\bal
\bQ_{1+}(\bI+\bQ_{+1}\bR_+\bQ_{1+})^{-1} \bQ_{1+} = (\bI+\bW_{1+}\bR_+)^{-1} \bW_{1+}
\eal
where $\bW_{1+}= \bQ_{1+}^2$ and adopting \eqref{eq.T.KKT.Rtil}, \eqref{eq.thm.R*.1}, one obtains
\bal
\label{eq.T.KKT.Rtil+}
\bR_+ = \bLam_{2+}^{-\frac{1}{2}}(\mu_2^{-1}\bI- \bLam_{2+}^{\frac{1}{2}}\bW_{1+}^{-1}\bLam_{2+}^{\frac{1}{2}})_+ \bLam_{2+}^{-\frac{1}{2}}
\eal
Since
\bal
P_T \ge tr(\bR) =tr(\tilde{\bR}) \ge tr(\bR_+),\ P_I \ge tr(\bW_2\bR) =tr(\bLam_2\tilde{\bR}) = tr(\bLam_2\bR_+)
\eal
one can set, without loss of optimality, $\bR_{22}=0$, $\bR_{12}=0$, $\bR_{21}=0$, and transform \eqref{eq.T.KKT.Rtil+} to
\bal
\label{eq.T.KKT.Rtil2}
\tilde{\bR} = (\bLam_{2}^{\dag})^{\frac{1}{2}}(\mu_2^{-1}\bI- \bLam_{2}^{\frac{1}{2}}\tilde{\bW}_{1}^{-1}\bLam_{2}^{\frac{1}{2}})_+ (\bLam_{2}^{\dag})^{\frac{1}{2}}
\eal
and hence, as desired,
\bal
\label{eq.T.KKT.R2}
\bR= \bU_2\tilde{\bR}\bU_2^+ = \bW_{\mu}^{\dag}(\bI- \bW_{\mu}\bW_{1}^{-1}\bW_{\mu})_+ \bW_{\mu}^{\dag}
\eal

\subsection{Proof of Proposition \ref{prop.C.inf}}

To prove the "if" part, observe that $\sN(\bW_2) \notin \sN(\bW_1)$ implies $\exists \bu: \bW_2\bu=0, \bW_1\bu \neq 0$. Now set $\bR=P_T\bu\bu^+$, for which $tr(\bR) =P_T, tr(\bW_2\bR)=0$, so it is feasible for any $P_T, P_I$. Furthermore,
\bal
C \ge C(\bR) = \log(1+P_T\bu^+\bW_1\bu) \rightarrow \infty
\eal
as $P_T \rightarrow \infty$, since $\bu^+\bW_1\bu >0$.

To prove the "only if" part, assume that $\sN(\bW_2) \in \sN(\bW_1)$. This implies that $\sR(\bW_1) \in \sR(\bW_2)$ (since $\sR(\bW)$ is the complement of $\sN(\bW)$ for Hermitian $\bW$). Let
\bal
\bW_k = \bU_{k+} \bLam_{k}\bU_{k+}^+,\ k=1,2
\eal
where $\bU_{k+}$ is a semi-unitary matrix of active eigenvectors of $\bW_k$ and diagonal matrix $\bLam_{k}$ collects its strictly-positive eigenvalues. Notice that, from the IPC,
\bal
P_I \ge tr(\bW_2\bR) = tr(\bLam_2\bU_{2+}^+\bR\bU_{2+}) \ge \gl_{r_2} tr(\bU_{2+}^+\bR\bU_{2+})
\eal
where $\gl_{r_2}>0$ is the smallest positive eigenvalue of $\bW_2$, so that
\bal
 \gl_1(\bU_{2+}^+\bR\bU_{2+}) \le P_I/\gl_{r_2} < \infty
\eal
for any $P_T$. On the other hand, $\sR(\bW_{1}) \in \sR(\bW_2)$ implies $span\{\bU_{1+}\} \in span\{\bU_{2+}\}$ and hence
\bal
 \gl_1(\bU_{1+}^+\bR\bU_{1+}) \le \gl_1(\bU_{2+}^+\bR\bU_{2+}) \le P_I/\gl_{r_2} < \infty
\eal
so that
\bal
C(P_T) &= \log|\bI+ \bLam_1\bU_{1+}^+\bR^*\bU_{1+}| = \sum_i \log(1+\gl_i(\bLam_1\bU_{1+}^+\bR^*\bU_{1+}))\\ \notag
&\le m \log(1+\gl_1(\bW_1)\gl_1(\bU_{1+}^+\bR^*\bU_{1+})) \le m \log(1+\gl_1(\bW_1)P_I/ \gl_{r_2}) < \infty
\eal
for any $P_T$, as required.

\vspace*{-1\baselineskip}
\subsection{Proof of Proposition \ref{prop.W2.r1}}

Start with the matrix inversion Lemma to obtain
\bal
(\mu_1\bI +\mu_2\gl_2\bu_2\bu_2^+)^{-1} = ((\mu_1 &+\mu_2\gl_2)^{-1}- \mu_1^{-1})\bu_2\bu_2^+ +\mu_1^{-1}\bI
\eal
so that \eqref{eq.prop.W2-r1.R*2} follows from \eqref{eq.R*.general}. Since $\bW_1$ is full-rank and $\bW_2$ is rank-1, it follows that the TPC is always active, $\mu_1>0$ and $tr(\bR)=P_T$, from which one obtains
\bal
\label{eq.prop.W2-r1.R*.4}
m\mu_1^{-1} -\alpha - tr(\bW^{-1}_1) = P_T
\eal
When the IPC is active, $tr(\bW_2\bR)=P_I$, it follows that
\bal
\label{eq.prop.W2-r1.R*.5}
\gl_2(\mu_1+\mu_2\gl_2)^{-1}= P_I +\gl_2\bu_2^+\bW^{-1}_1\bu_2
\eal
Solving \eqref{eq.prop.W2-r1.R*.4} and \eqref{eq.prop.W2-r1.R*.5} for $\mu_1$, one obtains 1st equality in \eqref{eq.prop.W2-r1.5}; using it in \eqref{eq.prop.W2-r1.R*.5} results in 2nd equality in \eqref{eq.prop.W2-r1.5}. \eqref{eq.prop.W2-r1.4} and 1st inequality in \eqref{eq.prop.W2-r1.3} ensure that $\bR^* >0$, since
\bal
\label{eq.prop.W2-r1.R*.6}
\mu_1^{-1} > \gl_1(\bW^{-1}_1) +\alpha \ge \gl_1(\bW^{-1}_1+\alpha\bu_2\bu_2^+)
\eal
where 1st inequality is due to 1st inequality in \eqref{eq.prop.W2-r1.3} and \eqref{eq.prop.W2-r1.R*.5} while 2nd inequality is from $\gl_1(\bA+\bB) \le \gl_1(\bA)+\gl_1(\bB)$ where $\bA,\ \bB$ are Hermitian matrices (see e.g. \cite{Horn-85}). It follows from \eqref{eq.prop.W2-r1.R*.6} that $\mu_1^{-1}\bI > \bW^{-1}_1 +\alpha\bu_2\bu_2^+$ and hence $\bR^*>0$, and that $\mu_1>0$, as required. 2nd inequality in \eqref{eq.prop.W2-r1.3} ensures that the IPC is active, $\mu_2>0$.

To obtain \eqref{eq.prop.W2-r1.R*1}, observe that $\bR_{WF}$ is feasible under \eqref{eq.prop.W2-r1.1} and \eqref{eq.prop.W2-r1.1a}:
\bal
tr(\bR_{WF}) =P_T,\ tr(\bW_2\bR_{WF}) \le P_I,\ \bR_{WF} >0.
\eal
Since it is a solution without the IPC (as the standard full-rank WF solution), it is also optimal under the IPC.

\vspace*{-.51\baselineskip}
\subsection{Proof of Proposition \ref{prop.rR}}

We consider first the case when $\bW_{\mu}$ is full-rank, i.e. when either the TPC is active, $\mu_1>0$, or/and $\bW_2>0$. It follows from \eqref{eq.T.KKT} that
\bal
(\bI+\bW_1\bR^*)^{-1}\bW_1\bR^* = \bW_{\mu}^2\bR^*
\eal
so that, since $(\bI+\bW_1\bR)$ and $\bW_{\mu}^2$ are full-rank,
\bal
r(\bR^*) &= r(\bW_{\mu}^2\bR^*) = r(\bW_1\bR^*) \le \min\{r(\bW_1),r(\bR^*)\} \le r(\bW_1)
\eal

The case of rank-deficient $\bW_{\mu}^2$ (i.e. when $\mu_1=0$ and $\bW_2$ is rank-deficient) is more involved. In this case, it follows from Proposition \ref{prop.TPC.inact} that $\mathcal{N}(\bW_2) \in \mathcal{N}(\bW_1)$ and hence $\mathcal{R}(\bW_1)\in \mathcal{R}(\bW_2)$ (if $\bW$ is Hermitian, $\mathcal{R}(\bW)$ is the complement of $\mathcal{N}(\bW)$), from which the following equivalency can be established, which is instrumental in the proof.

\begin{prop}
\label{prop.NW2}
If $\bW_2$ is rank-deficient and the TPC is redundant for the problem (P1) in \eqref{eq.C.def} under the constraint in \eqref{eq.SR.2}, then (P1) has the same value as the following problem (P2):
\bal
\label{eq.C.tilda}
(P2):\ \max_{\tilde{\bR} \ge 0} \tilde{C}(\tilde{\bR})\ \mbox{s.t.}\ tr(\tilde{\bLam}_2\tilde{\bR}) \le P_I,\ tr(\tilde{\bR}) \le P_T
\eal
where $\tilde{C}(\tilde{\bR})= |\bI+ \tilde{\bW}_1\tilde{\bR}|$, $\tilde{\bW}_1 = \bU_{2+}^+\bW_1\bU_{2+}$, $\tilde{\bLam}_2 = \bU_{2+}^+\bW_2\bU_{2+}>0$ is a diagonal matrix of strictly-positive eigenvalues of $\bW_2$ and $\bU_{2+}$ is a semi-unitary matrix whose columns are the corresponding active eigenvectors of $\bW_2$. Furthermore, an optimal covariance $\bR^*$ of (P1) can be expressed as follows:
\bal
\label{eq.R*2}
\bR^* = \bU_{2+}\tilde{\bR^*}\bU_{2+}^+
\eal
where $\tilde{\bR^*}$ is a solution of \eqref{eq.C.tilda}:
\bal
\label{eq.R.tild}
\tilde{\bR}^* = \tilde{\bLam}_2^{-\frac{1}{2}} (\mu_2^{-1}\bI-\tilde{\bLam}_2^{\frac{1}{2}}\tilde{\bW}_1^{-1} \tilde{\bLam}_2^{\frac{1}{2}})_+ \tilde{\bLam}_2^{-\frac{1}{2}}
\eal
and $\mu_2> 0$ is found from the IPC:
\bal
tr (\mu_2^{-1}\bI- \tilde{\bLam}_2^{\frac{1}{2}}\tilde{\bW}_1^{-1} \tilde{\bLam}_2^{\frac{1}{2}})_+ = P_I
\eal
\end{prop}
\begin{proof}
Let $\bR^*$ and $\tilde{\bR}^*$  be the solutions of (P1) and (P2) under the stated conditions and let $\bP_2=\bU_{2+}\bU_{2+}^+$ be a projection matrix on the space spanned by the active eigenvectors of $\bW_2$, i.e. on $\mathcal{R}(\bW_2)$. Note that $\bP_2\bW_k\bP_2=\bW_k$, $k=1,2$, since $\mathcal{R}(\bW_1)\in \mathcal{R}(\bW_2)$ under the stated conditions. Define $\tilde{\bR}'= \bU_{2+}^+\bR^*\bU_{2+}$ and observe that
\bal
P_T \ge tr(\bR^*) \ge tr(\tilde{\bR}'),\
P_I \ge tr(\bW_2\bR^*) = tr(\bP_2\bW_2\bP_2\bR^*) = tr(\tilde{\bLam}_2\tilde{\bR}')
\eal
so that $\tilde{\bR}'$ is feasible for (P2) and hence
\bal
\tilde{C}(\tilde{\bR}^*)\ge \tilde{C}(\tilde{\bR}')= \log|\bI+\tilde{\bW}_1\tilde{\bR}'| = \log|\bI+\bP_2\bW_1\bP_2\bR^*|
 =C(\bR^*)
\eal
On the other hand, let $\bR'= \bU_{2+}\tilde{\bR}^*\bU_{2+}^+$ and observe that
\bal
P_T \ge tr(\tilde{\bR}^*) = tr(\bR'),\
P_I \ge tr(\tilde{\bLam}_2\tilde{\bR}^*) = tr(\bP_2\bW_2\bP_2\bR') = tr(\bW_2\bR')
\eal
so that $\bR'$ is feasible for (P1) and hence
\bal
C(\bR^*)\ge C(\bR') =\log|\bI+\tilde{\bW}_1\tilde{\bR}^*|=\tilde{C}(\tilde{\bR}^*)
\eal
and finally $C(\bR^*)= \tilde{C}(\tilde{\bR}^*)$, $\mu_1=0$ (since the TPC is redundant for the original problem and hence for both problems) and the desired result follows.
\end{proof}

Note that Proposition \ref{prop.NW2} establishes the optimality of projecting all matrices on the sub-space $\mathcal{R}(\bW_2)$ and solving the projected problem instead, if the TPC is not active and $\bW_2$ is rank-deficient, i.e. if $\bW_{\mu}$ is rank-deficient.

Adopting the KKT condition in \eqref{eq.T.KKT} to the problem in \eqref{eq.C.tilda}, one obtains:
\bal
(\bI+\tilde{\bW}_1\tilde{\bR}^*)^{-1}\tilde{\bW}_1\tilde{\bR}^* = \mu_2\tilde{\bLam}_2\tilde{\bR}^*
\eal
so that
\bal
r(\tilde{\bR}^*) &= r(\tilde{\bLam}_2\tilde{\bR}^*) = r(\tilde{\bW}_1\tilde{\bR}^*) \le \min(r(\tilde{\bW}_1),r(\tilde{\bR}^*))
    \le r(\tilde{\bW}_1) \le r(\bW_1)
\eal
and, from \eqref{eq.R*2}, $r(\bR^*)=r(\tilde{\bR}^*)$, so that $r(\bR^*) \le r(\bW_1)$, as desired.

\subsection{Proof of Proposition \ref{prop.r1.IPC}}

To establish these results, we need the following technical Lemma, which can be established via the standard continuity argument.
\begin{lemma}
\label{lemma.I-W}
Let $\bW=\gl\bu\bu^+$ be rank-one positive semi-definite matrix, $\gl>0$. Then,
\bal
(\bI-\bW^{-1})_+=(1-\gl^{-1})_+\bu\bu^+
\eal
\end{lemma}
Note that the $(\cdot)_+$ operator eliminates all singular modes of $\bW$ and hence its singularity is not a problem, which is somewhat similar to using pseudo-inverse for a singular matrix.

To prove the 1st case, we assume that $\bW_2$ is not singular and discuss the singular case later. Setting $\bW=\bW_{2\mu}^{-\frac{1}{2}}\bW_1\bW_{2\mu}^{-\frac{1}{2}}$ and applying this Lemma to $(\bI-(\bW_{2\mu}^{-\frac{1}{2}}\bW_1\bW_{2\mu}^{-\frac{1}{2}})^{-1})_+$
in \eqref{eq.thm.R*.1}, one obtains $\bR^*$ as in \eqref{eq.prop.r1.IPC.1}, after some manipulations, with $\bW_2^\dag=\bW_2^{-1}$. The condition $\gamma_I< \gamma_1$ ensures that the TPC is redundant, so that $\mu_1=0$ and hence $\bW_{2\mu}=\mu_2\bW_2>0$, $tr(\bW_2\bR^*)=P_I$ (since the IPC is active).

If $\bW_2$ is singular and the TPC is redundant, then one can project all matrices on $\mathcal{R}(\bW_2)$ and solve the projected problem instead without loss of optimality, as was shown in Proposition \ref{prop.NW2}. After some manipulations, this can be shown to result in using the pseudo-inverse instead of the inverse of $\bW_2$.

To prove the $\gamma_I\ge \gamma_2$ case, note that, under this condition, $\bR^* = P_T\bu_1\bu_1^+$ is feasible under the joint constraint (TPC+IPC). Since it is also optimal without the IPC, it has to be optimal under the joint constraints as well. This proves the "if" part. To prove the "only if" (necessary) part, observe that if $P_T \bu_1^+\bW_2\bu_1 > P_I$, then $\bR^* = P_T\bu_1\bu_1^+$ is not feasible and hence cannot be optimal under the IPC.

To prove the last case, $\gamma_1 \le \gamma_I < \gamma_2$, use \eqref{eq.thm.R*.1} and note that both constraints are now active (since neither \eqref{eq.prop.r1.IPC.1} nor $\bR^* = P_T\bu_1\bu_1^+$ are feasible under the stated conditions). Applying Lemma \ref{lemma.I-W} as in the 1st case, one obtains \eqref{eq.prop.r1.IPC.3} after some manipulations.

\subsection{Proof of Proposition \ref{prop.f12.prop}}
The proof is based on the following technical lemma.
\begin{lemma}
For a fixed $\mu_2$, $P_1(\mu_1,\mu_2)$ is a decreasing function of $\mu_1$. Likewise, for a fixed $\mu_1$, $P_2(\mu_1,\mu_2)$ is a decreasing function of $\mu_2$.
\end{lemma}
\begin{proof}
Let us consider the capacity as a function of the total transmit power $P_1(\mu_1,\mu_2)=tr\bR^*(\mu_1,\mu_2)$: $C(P_1)$. It is straightforward to see that $C(P_1)$ is a concave function (see e.g. \cite{Boyd-04}, exercise 5.32), so that
\bal
\partial^2 C(P_1)/\partial P_1^2 \le 0
\eal
In addition,
\bal
\partial C(P_1)/\partial P_1 = \mu_1
\eal
Combining these,
\bal
\frac{\partial^2 C(P_1)}{\partial P_1^2}= \frac{\partial\mu_1}{\partial P_1} \le 0
\eal
so that $\partial P_1/\partial\mu_1 \le 0$, as required.

The inequality $\partial P_2/\partial\mu_2 \le 0$ is proved in a similar way.
\end{proof}

Since $f_1(\mu_{10},\mu_2)=0$, it follows that either $\mu_{10}=0$ or $P_1(\mu_{10},\mu_2)=P_T$. In the latter case, $P_1(\mu_{1},\mu_2)\le P_T$ (from the above Lemma) and hence $f_1(\mu_{1},\mu_2)\le 0$ for any $\mu_1>\mu_{10}$, and the opposite inequalities hold for $\mu_1<\mu_{10}$. If $\mu_{10}=0$, then $P_1(\mu_{10},\mu_2)\le P_T$ and hence $f_1(\mu_{1},\mu_2)\le 0$ for any $\mu_1>\mu_{10}=0$ from the above Lemma. This proves the desired property of $f_1(\mu_1,\mu_2)$. The same property of $f_2(\mu_1,\mu_2)$ is proved in a similar way.

\subsection{Proof of Proposition \ref{prop.IBA.conv}}
First, we proof that the sequences $\{\mu_{1k}\}_{k=1}^{\infty}$ and $\{\mu_{2k}\}_{k=1}^{\infty}$ generated by the IBA are decreasing and increasing respectively, if $P_{1(2)}(\mu_{1},\mu_{2})$ are decreasing functions. Since $\mu_1, \mu_2$ are also bounded from below and above (as in \eqref{eq.mu.l.u}), this will ensure the required convergence property.

Consider iteration 1 of the IBA, $k=1$. Step 2 ensures that $f_1(\mu_{11},0)=0$, $\mu_{11}>0$ and hence $P_1(\mu_{11},0)=P_T$. If $P_2(\mu_{11},0)\le P_I$, then $(\mu_{11},0)$ is optimal (the IPC is inactive) and the algorithm terminates. Otherwise, $P_2(\mu_{11},0)> P_I$ and step 3 ensures that $f_2(\mu_{11},\mu_{21})=0$ with $\mu_{21}>\mu_{20}=0$ and hence $P_2(\mu_{11},\mu_{21})= P_I$.

Consider now iteration $k=2$. If $P_1(\mu_{11},\mu_{21}) = P_T$, then $(\mu_{11},\mu_{21})$ is optimal and the IBA terminates. Otherwise, $P_1(\mu_{11},\mu_{21}) < P_T$ (due to the decreasing property of $P_1$ and $\mu_{21}>0$) and step 2 ensures that $f_1(\mu_{12},\mu_{21})=0$ so that either $\mu_{12}=0$ (inactive TPC) or $P_1(\mu_{12},\mu_{21}) = P_T$ and, in both cases, $\mu_{12} < \mu_{11}$. If $P_2(\mu_{12},\mu_{21}) = P_I$, then, at step 3, $\mu_{22}=\mu_{21}$, $(\mu_{12},\mu_{22})$ is optimal and the IBA terminates. Otherwise, $P_2(\mu_{12},\mu_{21}) > P_I$, step 3 ensures that $f_2(\mu_{12},\mu_{22})=0$ so that $P_2(\mu_{12},\mu_{22})= P_I$ and hence $\mu_{22}>\mu_{21}$.

Continuing this indefinitely, one obtains
\bal\notag
&\mu_{11}\ge \mu_{12}\ge..\ge\mu_{1k}\ge...\\
&\mu_{21}\le \mu_{22}\le..\le\mu_{2k}\le...
\eal
as required. Also, if $\mu_{1k}=\mu_1^*$ at some step $k$, then $\mu_{2k}=\mu_2^*$ and the IBA terminates at this optimal point. Likewise, if $\mu_{2k}=\mu_2^*$ for some $k$, then $\mu_{1(k+1)}=\mu_1^*$, $\mu_{2(k+1)}=\mu_2^*$ and the IBA terminates.

It was shown above that the IBA terminates in one iteration under condition 1. We now verify that $P_{1(2)}(\mu_{1},\mu_{2})$ are decreasing functions under conditions 2 and 3. To this end, let $\gl_{1i}=\gl_i(\bW_1)$, $\gl_{2i}=\gl_i(\bW_2)$. When $\bW_1$ and $\bW_2$ have the same eigenvectors, $\bR^*$ has the same eigenvectors too. Hence,
\bal
P_{1}(\mu_{1},\mu_{2}) &= \sum_i ((\mu_1+\mu_2\gl_{2i})^{-1}-\gl_{1i}^{-1})_+\\
P_{2}(\mu_{1},\mu_{2}) &= \sum_i \gl_{2i}((\mu_1+\mu_2\gl_{2i})^{-1}-\gl_{1i}^{-1})_+
\eal
which are clearly decreasing functions of $\mu_{1},\mu_{2}$.

When $\bR^*(\mu_1,\mu_2)$ is full-rank, the $(\cdot)_+$ operator is redundant (all eigenmodes are active) and hence
\bal
\bR^* = \bW_{\mu}^{-1} - \bW_1^{-1}
\eal
so that
\bal
P_{1}(\mu_{1},\mu_{2}) &= \sum_i (\mu_1+\mu_2\gl_{2i})^{-1}-tr\bW_1^{-1}\\
P_{2}(\mu_{1},\mu_{2}) &= \sum_i \gl_{2i}(\mu_1+\mu_2\gl_{2i})^{-1}-tr\bW_2\bW_1^{-1}
\eal
which are clearly decreasing functions of $\mu_{1},\mu_{2}$.

\end{document}